%% file: icfp-paper.tex
\documentclass[acmsmall]{acmart}

\input{prelude.tex}
\input{examples.tex}
\renewcommand\hscodestyle{%
   \setlength\leftskip{0.5cm}%
}

\acmJournal{PACMPL}
\acmVolume{1}
\acmNumber{ICFP} 
\acmArticle{1}
\acmYear{2018}
\acmMonth{1}
\acmDOI{} 
\startPage{1}

\setcopyright{none}

\usepackage{booktabs}   
\usepackage{subcaption} 

\usepackage{fancyvrb}
\begin{document}

\title[Simply RaTT]{Simply RaTT}

\subtitle{A Fitch-style Modal Calculus for Reactive Programming
  Without Space Leaks}                     



\author{Patrick Bahr}
\affiliation{
  \institution{IT University of Copenhagen}            
}
\email{paba@itu.dk}          

\author{Christian Uldal Graulund}
\affiliation{
  \institution{IT University of Copenhagen}           
}
\email{cgra@itu.dk}         

\author{Rasmus Ejlers M{\o}gelberg}
\affiliation{
  \institution{IT University of Copenhagen}           
}
\email{mogel@itu.dk}         

\begin{abstract}
Functional reactive programming (FRP) is a paradigm for programming with
signals and events, allowing the user to describe reactive programs 
on a high level of abstraction. For this to make sense, an FRP language 
must ensure that all programs are causal, and can be implemented without
introducing space leaks and time leaks. To this end, some FRP languages
do not give direct access to signals, but just to signal functions. 

Recently, modal types have been suggested as an alternative approach
to ensuring causality in FRP languages in the synchronous case, 
giving direct access to the signal and event
abstractions. This paper presents \emph{Simply RaTT}, a new modal
calculus for reactive programming. Unlike prior calculi, Simply RaTT
uses a Fitch-style approach to modal types, which simplifies the type
system and makes programs more concise. Echoing a
previous result by Krishnaswami for a different language, we devise an
operational semantics that safely executes Simply RaTT programs
without space leaks.

We also identify a source of time leaks present in other modal FRP
languages: The unfolding of fixed points in delayed computations.
The Fitch-style presentation allows an easy way to rules out these
leaks, which appears not to be possible in the more traditional 
dual context approach. 
\end{abstract}



\maketitle
%

\section{Introduction}
\label{sec:introduction}

Reactive programs are programs that engage in an ongoing dialogue with their environment, 
taking inputs and producing outputs, typically dependent on an internal state. Examples include
GUIs, servers, and control software for components in cars, aircraft, and robots.
These are traditionally
implemented in imperative programming languages using often complex webs of components 
communicating through callbacks and shared state. As a consequence, reactive programming in imperative 
languages is error-prone and program behaviour difficult to reason about. This is unfortunate since
many of the most safety-critical programs in use today are reactive. 

The idea of Functional Reactive Programming (FRP)~\cite{FRAN} is to
bring reactive programming into the functional paradigm by providing
the programmer with abstractions for describing the dataflow between
components in a simple and direct way. At the same time, this should
give the usual benefits of functional programming: Modular programming using
higher-order functions, and simple equational reasoning.
%
The abstractions
provided by the early FRP languages were \emph{signals} and
\emph{events}: A signal of type $A$ is a time-varying value of type
$A$, and an event of type $A$ is a value of type $A$ appearing at some
point in time.  The notion of time is abstract, but can, depending on
the application, be thought of as either discrete or continuous.

For such high-level abstractions to make sense, the language designer must ensure that all programs 
can be executed in an efficient way. A first problem is ensuring \emph{causality}, i.e., the property that
the value of output signals at a given time only depends on the values read from input signals before
or at that time. For example, implementing signals in the discrete time case simply as streams will break this
abstraction, as there are many non-causal functions from streams to streams. Another issue is 
\emph{time leaks}, i.e., the problem of programs exhibiting gradually slower response time, typically due
to intermediate values being recomputed whenever output is needed. The related notion of 
\emph{space leaks} is the problem of programs holding on to memory while continually allocating more 
until they eventually run out of memory. 

A good language for FRP should only allow programmers to write causal functions. On the other hand, 
in expressive programming languages some of the responsibility for avoiding the problems of space 
and time leaks must be left to the programmer. For example, if the language has linked lists, a programmer 
could write a function that stores all input in a list, leading to a space leak. We will refer to this as an 
\emph{explicit} space leak, since it can be detected from the code. A good FRP language should
avoid \emph{implicit} space and time leaks, i.e., leaks that are caused by the language implementation, and 
so are out of the programmers control. 

Due to these concerns, newer libraries and languages for FRP do not
give the programmer direct access to signal and event types. For
example, Arrowised FRP~\cite{nilsson2002} has a primitive notion of
signal functions and provides combinators for combining these to
construct dataflow networks statically, along with switching operators
for dynamically changing these networks. This approach sacrifices some
of the simplicity and flexibility of the original suggestions for FRP,
and the switching combinators have an ad hoc flavour. Moreover, to the
best of our knowledge, no strong guarantees concerning space or time
leaks have been proved in this setting.

\subsection{Modal FRP calculi}

Recently, a number of authors 
(\cite{jeltsch2013temporal,jeffrey2012,jeffrey2014,krishnaswami2011ultrametric,krishnaswami13frp}) 
have suggested using modal types for functional reactive programming. These all work in the synchronous
case of time being given by a global clock.
With this assumption, the resulting languages can be thought of as extensions of 
synchronous dataflow languages such as  
Lustre~\cite{caspi1987lustre}, and Lucid Synchrone~\cite{pouzet2006lucid} with higher-order functions 
and operations for dynamically changing the dataflow network. This restricted setting covers many applications
of FRP, and in this paper we shall restrict ourselves to that as well. Since continuous time can be simulated
by discrete time (see \autoref{sec:related-work}), we will further restrict ourselves to discrete time. 

Under the assumption of a global discrete clock 
a signal is simply a stream. Causality is ensured by using a modal type operator
$\Delay$ to encode the notion of a time step in types: A value of type $\Delay A$ is a computation
returning a value of type $A$ in the next time step. Using $\Delay$, one can describe the streams corresponding
to signals as a type satisfying the type isomorphism $\Str{A} \cong A \times
\Delay\Str{A}$, capturing the fact that the tail of the stream is only available in the next time step. 
Streams and programs processing streams can be defined recursively using the guarded
fixed point combinator of~\citet{nakano2000} taking input of type $\Delay A \to A$ and producing elements of type $A$ as output. 

The most advanced programming language of this kind, in terms of operational semantics with
run-time guarantees, is that of~\citet{krishnaswami13frp}. This language extends the simply typed 
lambda calculus with two modal type operators: The $\Delay$ mentioned above, as well as one 
for classifying stable, i.e., time-invariant data. Unlike the arrowised approach to FRP, Krishnaswami's 
calculus gives direct access to streams as a data type which can even be nested to give streams of 
streams. Other important data types, such as events can be
encoded using \emph{guarded recursive types}, a concept also stemming from~\citet{nakano2000}.

Krishnaswami's calculus has an operational semantics for 
evaluating terms in each step of the global clock, and this can be extended to a step-by-step 
evaluation of streams. The language is total in the sense that each step evaluates to a value in finite
time (a property often referred to as \emph{productivity}). 
The operational semantics evaluates by storing delayed computations on a heap, 
and Krishnaswami shows that all heap data can be 
safely garbage collected after each evaluation step, effectively guaranteeing the absence of (implicit) space leaks. 

\subsection{Fitch-style modal calculi}

Like most modal calculi, Krishnaswami's calculus uses let-expressions to program with modalities. 
This affects the programming style: Many programs consist of a long series of unpacking statements, essentially giving
access to the values produced by delayed computations in the next time step, 
followed by relatively short expressions manipulating these.
While this can be to a large extent be dealt with using syntactic sugar, it has a more fundamental problem, which is
harder to deal with: It complicates equational reasoning about programs. This is an important issue, since simple 
equational reasoning is supposed to be one of the benefits of functional programming. Our long-term goal is to design
a dependent type theory for reactive programming in which programs have operational guarantees like the ones
proved by Krishnaswami, and where program specifications can be expressed using dependent types. Introducing 
let-expressions in terms will lead to let-expressions also in the types, which is a severe complication of the type theory. 

Fitch-style modal calculi~\cite{Fitch:Symbolic,clouston2018fitch} are an alternative approach to modal types not using
let-expressions. Instead, elements of modal types are constructed by abstracting tokens from the context, 
and modal operators are likewise eliminated by placing tokens in the context. Recent research in 
guarded dependent type theory~\cite{bahr2017clocks,clouston2018modal} 
has shown the benefit of this approach also for dependent types. Guarded
dependent type theory is an extension of Martin-L{\"o}f type theory~\cite{martin1984intuitionistic} 
with a delay modality reminiscent of the $\Delay$
used in modal FRP together with Nakano's fixed point combinator also mentioned above. In this setting, the token used
in the Fitch-style approach is thought of as a 'tick' -- evidence that time has passed -- which can be used to open a
delayed computation. Using ticks, one can prove properties of guarded recursive programs in a compellingly
simple way.

\subsection{Simply RaTT}

In this paper, we present Simply Typed Reactive Type Theory (Simply RaTT) a simply typed 
calculus for reactive programming based on the Fitch-style approach to modal types. This is a first step towards our goal 
of a dependently typed theory for reactive programming (RaTT), but already the simply typed version offers several benefits
over existing approaches. Compared to Krishnaswami's calculus, Simply RaTT has a
significantly simpler type system. The Fitch-style approach eliminates the need for the qualifiers
'$\sym{now}$', '$\sym{later}$' and '$\sym{stable}$' used in Krishnaswami's calculus
on variables and term judgements. Similar (but not quite the same) qualifiers can be derived from 
the position of variables relative to tokens in contexts in Simply RaTT. 
Moreover, we eliminate the need for allocation tokens, 
a technical tool used by Krishnaswami to control heap allocation. This, together with the Fitch-style 
typing rules makes programs shorter and (we believe) more readable than in Krishnaswami's 
calculus. 

%
%
%

Compared to the standard approach to modal types, the Fitch-style used here is based on a shift in time-dependence. Whereas
terms in Krishnaswami's language can look into the future (since now-terms can depend on later-variables), terms in
Simply RaTT can only look into the past (since later-expressions can depend on now-variables). This explains how let-expressions
are eliminated: There is no need to refer to the values produced in the future by delayed computations. Instead, Simply RaTT allows
delayed computations from the past to be run in the present. 

We prove a garbage collection result similar to that proved by Krishnaswami, and show
how this can be used to construct a safe evaluation strategy for stream transducers written 
in our language. Input to stream transducers are treated as delayed computations, and therefore
stored in a heap and garbage collected in the next time step. 

We also identify and eliminate a source of time leaks present in previous approaches. This 
is best illustrated by the following two implementations of the stream of natural numbers written
in Haskell-notation:
\begin{center}
  \vspace{-1em}
  \begin{minipage}[t]{0.45\linewidth}
    \haskellLeakyNats
  \end{minipage}
  \begin{minipage}[t]{0.45\linewidth}
    \haskellSafeNats
  \end{minipage}
\end{center}
On most machines (some compilers
may use clever techniques to detect this problem), the evaluation of the $n$th element of 
$\Varid{leakyNats}$ will not use the previously computed values, but instead compute it using $n$ 
successive applications of $\Varid{suc}$, resulting in a time leak. This is indeed what happens on 
Krishnaswami's machine and also the machine of this paper. Contrary to that, the $\Varid{nats}$ 
example uses an internal state declared explicitly in the type of $\Varid{from}$ to 
maintain a constant evaluation time for each step. 
In this paper we identify the source of the time leak to be the ability to unfold fixed points in delayed
computations, and use this to eliminate examples such as
$\Varid{leakyNats}$ in typing. The ability to control when unfolding
of fixed points are allowed relies crucially on the Fitch-style
presentation, and it is very unclear whether a similar restriction can
be added to the traditional dual context presentation. 

The calculus is illustrated through examples showing how to  implement a small FRP library as well
as how to simulate the most basic constructions of Lustre in Simply RaTT. Examples are also used
to illustrate our abstract machine for evaluating streams and stream transducers. 

\subsection{Overview of paper}

The paper is organised as follows: \autoref{sec:overview-language} gives an overview of the language
introducing the main concepts and their intuitions through examples. \autoref{sec:oper-semant} defines the 
operational semantics, including the evaluation of stream transducers and states the garbage collection results for 
these. \autoref{sec:stream_library} shows how to implement a small library for reactive programming in Simply
RaTT and \autoref{sec:lustre} shows how to encode the most basic constructions of the synchronous dataflow language Lustre in Simply RaTT. 
\autoref{sec:metatheory} sketches the proof of our garbage collection result. Finally, \autoref{sec:related-work} 
describes related work and \autoref{sec:concl-future-work} concludes and describes future work.

\section{Simply RaTT}
\label{sec:overview-language}

This section gives an overview of the Simply RaTT language. 
The complete formal description of the syntax of the language, 
and in particular the typing rules, can be found in \autoref{fig:syntax}, \autoref{fig:types},
and \autoref{fig:typing}.

The type system of Simply RaTT extends that of the simply typed 
lambda calculus with
two modal type operators: $\Delay$ for classifying delayed
computations, and $\Box$ for classifying stable computations, i.e.,
computations that can be performed safely at any time in the
future. We start by describing the constructions for $\Delay$.

Data of type $\Delay A$ are \emph{computations} that produce data of
type $A$ 
in the next time step. To perform such a computation we must wait a
time step, as represented in typing judgements by the addition of a
$\tick$ (pronounced 'tick') in the context. More precisely, the typing
rule for eliminating $\Delay$ states that if
$\hastype{\Gamma}{t}{\Delay A}$ then
$\hastype{\Gamma,\tick,\Gamma'}{\adv(t)}{A}$. The $\tick$ in the
context of $\adv(t)$ should be thought of as separating variables in
time: Those in $\Gamma$ are available one time step before those in
$\Gamma'$.  Since there can be at most one $\tick$ in a context, we
will refer to these times as 'now' and 'later'. The typing assumption
on $t$ states that it has type $\Delay A$ now, and the conclusion
states that $\adv(t)$ has type $A$ later. The constructor for
$\Delay A$ states that if $\hastype{\Gamma,\tick}{t}{A}$, i.e., if $t$
has type $A$ later, but depends only on variables available now, then
it can be turned into a \emph{thunk} $\delay(t)$ of type $\Delay A$
now.

Note that terms in `later' judgements can refer to variables available
now as well as later, but 'now' judgements can only refer to variables
available now. This separates the Fitch-style approach of Simply RaTT
from the traditional dual context approach to calculi with modalities,
such as \citepos{krishnaswami13frp} 
modal calculus for reactive
programming. The latter also has a distinction between `later' and
`now', but the time dependencies work the opposite way: A
later-judgement can only depend on later-variables, whereas a
now-judgement can depend on both now- and later-variables.

Data of type $\Box A$ are time invariant \emph{computations} that
produce data of type $A$. That is, these computations can be executed
safely at any time in the future. To allow time invariant computations
to depend on initial data, that is, data available before the reactive program starts 
executing, contexts may contain a $\lock$
separating the context into \emph{initial
variables} (those to the left of $\lock$) and \emph{temporal variables} to the
right of $\lock$. There can be at most one $\lock$ in a context, and
$\Gamma, \tick$ is only well-formed if there is a $\lock$ in
$\Gamma$. Thus $\tick$ separates the temporal variables into now and
later.  We refer collectively to $\tick$ and $\lock$ as \emph{tokens}. Judgements
in a token-free context is referred to as an initial judgement. The three
kinds of judgements are summarised in \autoref{fig:judgement}.

\begin{figure}
  \begin{subfigure}[b]{0.3\linewidth}
    \[\Gamma \vdash t : A\]
    \caption{Initial judgement}
  \end{subfigure}
  \begin{subfigure}[b]{0.3\linewidth}
    \[\Gamma,\lock,\Gamma_N \vdash t : A\]
    \caption{Now judgement}
  \end{subfigure}
  \begin{subfigure}[b]{0.3\linewidth}
    \[\Gamma,\lock,\Gamma_N,\tick,\Gamma_L \vdash t : A\]
    \caption{Later judgement}
  \end{subfigure}
  \caption{The different type judgement forms. In these, the contexts $\Gamma$, $\Gamma_N$ and $\Gamma_L$ are assumed to be 
  token-free and contain variables referred to as initial variables,
  now-variables and, later-variables.}
  \label{fig:judgement}
\end{figure}

If $\hastype{\Gamma, \lock}t{A}$ then $t$ does not depend on any
temporal data, and can thus be thunked to a time invariant computation
$\hastype{\Gamma}{\rbox (t)}{\Box A}$
to be run at a later time.
The typing rule for eliminating $\Box$ states that if
$\hastype{\Gamma}{t}{\Box A}$  and $\Gamma'$ is token-free, then
$\hastype{\Gamma,\lock,\Gamma'}{\unbox(t)}{A}$. 
The restriction on $\Gamma'$ means that we can only run the time invariant
computation $t$ now, not later. 
This may seem to contradict the intuition for $\Box A$ given above,
but is needed to rule out certain time leaks as we shall
see below. Time invariant computations can still be run at arbitrary times 
in the future through the use of fixed points.

Both these modal type operators have restricted forms of applicative actions. 
In the case of $\Delay$, if $\hastype\Gamma t{\Delay (A \to B)}$ and 
$\hastype{\Gamma}u{\Delay A}$ then $\hastype\Gamma{t \delayapp u}{\Delay B}$ 
is defined as 
\[
t \delayapp u = \delay(\adv(t) (\adv(u))). 
\]
Note that this is only well-typed
if $\Gamma$ contains $\lock$ but not $\tick$, since the subterm
$\adv(f) (\adv(x))$ must be typed in context $\Gamma, \tick$,
and by the restrictions mentioned above, this is only a well-formed context
if $\lock$ is the only token in $\Gamma$. Similarly,
if $\hastype\Gamma t{\Box (A \to B)}$ and 
$\hastype{\Gamma}u{\Box A}$ then $\hastype\Gamma{t \boxapp u}{\Box B}$ is
defined as $\rbox(\unbox(t) (\unbox(u)))$. As above, this is only well-typed if $\Gamma$
is token-free. Note that neither $\Box$ nor $\Delay$ are applicative
functors in the 
sense of~\citet{mcbride2008applicative}, since there are generally no maps $A \to \Box A$,
nor $A \to \Delay A$. The former would force computations to be stable, and the latter
would push data into the future, which is generally unsafe as it can lead to space leaks.
This restriction is enforced in the type theory in the variable introduction rule, which does not allow
variables to be introduced over tokens. As a consequence, weakening of typing judgements 
with tokens is not admissible. An exception to this exists for the \emph{stable} types, 
as we shall see below. 

\subsection{Fixed points}
\label{sec:fixed-points}

Reactive programs can be defined recursively using a fixed point combinator.
To ensure productivity and causality, the recursion variable must be a delayed computation.
Precisely, the rule for fixed points state that if ${\hastype{\Gamma,\lock,x : \Delay A}{t}{A}}$ then
${\hastype{\Gamma}{\fix \; x.t}{\Box A}}$. These guarded recursive fixed points can be used to 
program with guarded recursive types such as guarded recursive streams
$\Str{A}$ satisfying the type isomorphism $\Str{A} \cong A \times
\Delay\Str{A}$. Terms of this type compute to an element in $A$ (the head) now, and a 
delayed computation of a tail. We will use $\mathbin{::}$ as infix notation for the right to left
direction of the isomorphism, i.e., $t\mathbin{::}u$ is a shorthand
for $\into\pair t u$. Given $t: A$ and $u : \Delay \Str{A}$, we thus
have $t\mathbin{::}u : \Str A$. 

As a simple example of a recursive definition, the stream of all zeros can be defined as
\zerosExample
Note that fixed points are time invariant in the sense of having a type of the form 
$\Box A$. This is because they essentially need to call themselves in the future. For this
reason, their definition cannot depend on temporal data, as can be seen from the typing rule, since
$x$ must be the only temporal variable in $t$.

As a second example of a recursively defined function, we define a map function for guarded streams.
This should take a function $A \to B$ as input and a stream of type $\Str{A}$ and produce a stream
of type $\Str{B}$. Since the input function will be called repeatedly at all future time steps it needs to
be time-invariant, and can be defined as:
  \mapExampleUgly
where $\head$ and $\tail$ compute the head and the tail of a stream,
respectively.

For readability we introduce the following syntax for defining fixed
points such as $\Varid{map}$:
  \mapExample
This should be read as defining the term to the left of $\lock$ as a
fixed point and in particular it allows us to write pattern matching
in a simple way. When type checking the right-hand side of this
definition, $\Varid{map}\, f$ should be given type
$\Delay(\Str{A} \to \Str{B})$ because it represents the recursion
variable. Any such definition can be translated syntactically to our
core language in a straightforward manner: Pattern matching is
translated to the corresponding elimination forms ($\pi_i$,
$\sym{case}$, $\out$) and the recursion syntax with $\lock$ is
translated to $\fix$. 

The type of guarded streams defined above is just one example of a guarded recursive type. Simply RaTT includes 
a construction for general recursive types $\mu\alpha.A$ satisfying type isomorphisms of the form
$\mu\alpha.A \cong A[\Delay(\mu\alpha.A)/\alpha]$. In these $\alpha$ can appear everywhere in $A$, including 
non-strictly positive and even negative positions. 
Another example of a guarded recursive type is that of events defined as 
$\Ev{A} = \mu \alpha . A + \alpha$,  and thus
satisfying $\Ev{A} \cong A + \Delay \Ev{A}$. 
Streams and events form the building blocks of functional reactive programming. 
Similarly to streams, one can define a map function for events using fixed points as follows
  \genMapE
where we write $\sym{val}\,t$ and $\sym{wait}\,t$ instead of
$\into\,(\interm_1\,t)$ and $\into\,(\interm_2\,t)$, respectively.


\subsection{Stable types}
\label{sec:stable-types}

Next we show how to define the stream of natural numbers using a helper function
mapping a natural number $n$ to the stream $(n, n+1 , n + 2,
\dots)$. A first attempt at defining $\Varid{from}$ could look as follows: 
  \fromExampleBad
is not well typed, because to type $\delay(n+1)$ the term
$n+1$ must have type $\Nat$ \emph{later}, but $n$ is a \emph{now}-variable. The number $n$ therefore must
be kept for the next time step, an operation that generally is unsafe, because general values can have references to
temporal data. For example, a value of type $\Delay\Str{A}$ in our machine
is a reference to the tail of a stream, which could be an input stream. Allowing such values to
be kept for the next step can lead the machine to store input data indefinitely, causing space leaks.
Similarly, values of function types
can contain references to time dependent data in closures and should therefore not be kept. 
On the other hand, a value of type natural numbers cannot contain
such references and so can safely be kept for the next time step. We say that $\Nat$ is a \emph{stable} type, and a 
grammar for these stable types is given in \autoref{fig:types}. Data of stable type can be kept one time step using the construction 
$\progress$ which 
allows a now-judgement of the form ${\hastype{\Gamma}{t}{A}}$ to be transformed to a later judgement of the form  
${\hastype{\Gamma,\tick,\Gamma'}{\progress\,t}{A}}$ if $\Gamma$ contains a $\lock$ and no $\tick$ and if $A$ is stable. In 
our operational semantics, $\progress\,t$ evaluates by evaluating $t$ to a value now pushing the result to the future. 
Postponing the evaluation of $t$ would be unsafe, since terms of stable types, unlike values of stable types, 
can refer to temporal data. 
Similarly, $\promote$ can be used to make stable initial data available in temporal judgements. 

We introduce the constructions $\progressapp$, defined as $t \progressapp u = \delay(\adv(t) (\progress\, u))$, and $\promoteapp$, defined as $t \promoteapp u = \rbox(\unbox(t) (\promote\, u))$,  with derived typing rules
\begin{mathpar}
      \inferrule*
      {\hastype{\Gamma}{t}{\Delay(A \to B)} \\
        \hastype{\Gamma}{u}{A} \\ \wfcxt{\Gamma,\tick} \\ 
        A \; \stable }
      {\hastype{\Gamma}{t \progressapp u}{\Delay B}}
      \and
       \inferrule*
      {\hastype{\Gamma}{t}{\Box(A \to B)} \\
        \hastype{\Gamma}{u}{A} \\ \wfcxt{\Gamma,\lock} \\
        A \; \stable }
      {\hastype{\Gamma}{t \promoteapp u}{\Box B}}
\end{mathpar}
Using this, $\Varid{from}$ and $\Varid{nats}$ can be defined as follows\\
\begin{minipage}{0.5\linewidth}
   \fromExample
\end{minipage}
\begin{minipage}{0.5\linewidth}
    \natsExample
\end{minipage}


Many programming languages would also allow $\Varid{nats}$ to be defined directly as fixed point as 
$\Varid{leakyNats}\mathrel{=}\mathrm{0}\mathbin{::}\Varid{map}\;(\mathbin{+}\mathrm{1})\;\Varid{leakyNats}{}$.
In Simply RaTT, however, such a definition would not be well typed, because the term $\Varid{map}(\rbox(+1))$
of type $\Box(\Str{\Nat} \to \Str{\Nat})$ would have to be unboxed in a context with a $\tick$
in order to type a term like
\[
\Varid{leakyNats} \, \lock = 0 \mathbin{::} \delay(\unbox(\Varid{map}(\rbox(+1)))) \delayapp \Varid{leakyNats}
\]
and this is not allowed according to the typing rule for $\unbox$.
We believe such a definition should be ruled out
because it leads to time leaks as explained in the introduction.
This indeed happens on the machined described in \autoref{sec:oper-semant} as well as the machine 
of \citet{krishnaswami13frp}. 

The time leak in the $\Varid{leakyNats}$ example above happens because
the fixed point definition of $\Varid{map}$ is unfolded in a delayed
term, allowing the term to be evaluated to grow for each iteration. In
the $\Varid{nats}$ example, on the other hand, the recursive
definition uses a state, namely the input to $\Varid{from}$, to avoid
repeating computations. Moreover, this state usage is essentially
declared in the type of $\mathit{from}$.

For similar reasons, the $\mathit{scary\_const}$ example of
\citet{krishnaswami13frp} in which all data from an input stream is
kept indefinitely by explicitly storing it in a stream of streams
cannot be typed in Simply RaTT. An implementation of
$\mathit{scary\_const}$ in Simply RaTT would require an explicit state
that stores all previous elements from the input stream. That could be
achieved by extending the language to include a list type
$\sym{List}\, A$, and defining that $\sym{List}\,A$ is stable if $A$
is. The fact that the memory usage of $\mathit{scary\_const}$ is
unbounded is then reflected by the fact that the state of type
$\sym{List}\, A$ that is needed for $\mathit{scary\_const}$ is unbounded
in size.

Note that we make crucial use of the Fitch-style presentation to rule
out $\Varid{leakyNats}$. In the more traditional dual context approach of
\citet{krishnaswami13frp}, it does not seem possible to have a similar
restriction on unfolding of fixed points. The difference is that
Simply RaTT ``remembers'' when we are under a delay whereas that
information is lost in the system by \citet{krishnaswami13frp}.
In the example of $\Varid{leakyNats}$, the leak stems from the call of 
$\Varid{map}$ in the tail, which in Krishnaswami's system is typed as a
regular now judgement, and thus cannot be prevented.


\subsection{Function types}
\label{sec:function-types}

The operational
semantics of Simply RaTT uses a heap for delayed computations as well as input streams. The 
operation $\delay(t)$ stores the computation $t$ on the heap and $\adv$ 
retrieves a delayed computation from the heap and evaluates it.
In this sense, $\delay$ and $\adv$ can be understood as computational effects. 


Our main result (\autoref{thr:lrl}) states that delayed computations and input data on the heap can be safely 
garbage collected after each computation step. This result relies crucially on the property that
open terms typed in now-judgements cannot retrieve delayed computations from the heap. 
One reason for this is that such terms can not  
contain $\adv$ unless under $\delay$. To maintain this invariant also for function calls, 
function types $A \to B$ are restricted
to functions with no retrieve effects. For this reason, functions may not
be constructed in a later-judgement. Later-variables can still be 
used for case-expressions, 
and so are included in Simply RaTT. The language could be 
extended with an extra function type with read effects, constructed by abstracting later-variables in 
later-judgements, but we found no use for this in our examples. 

For example, we can define a function reading an input stream and returning a stream of functions
as follows 
  \incrExample
and apply streams of functions as follows
  \streamAppExample

On the other hand, allowing lambda abstraction of later-variables
would type the following (rather contrived) stream definition
$\Varid{leaky}$, which breaks the safety of the garbage collection
strategy:
  \leakyExample
In particular, in the definition of $\Varid{leaky'}$ we use $\adv$ on
the recursive call of $\Varid{leaky'}$ inside a function. The problem
here is that the function body only gets evaluated when applied to an
argument. However, this application happens too late -- at a time
where the recursive call to $\Varid{leaky'}$ has already been garbage
collected (cf.\ \autoref{sec:counterexamples}).


\begin{figure}
  \[
    \arraycolsep=1.4pt
    \begin{array}{l@{\quad}lcl}
    \text{Types} &A,B &::=\; &A \mid  \Unit \mid \Nat \mid A \times B \mid A + B \mid A \to B \mid \Delay A \mid \Box A \mid \mu \alpha. A\\
    \text{Values} &v,w &::=\; &\unit \mid \bar{n} \mid \lambda x.t \mid
                              \pair{v}{w} \mid \interm_i\, v \mid \; \rbox\,t \mid \into\,v \mid  \fix  \; x.t \mid l\\
  \text{Terms} &s,t &::=\; &\unit \mid \bar{n} \mid \lambda x.t \mid
                              \pair{s}{t} \mid \interm_i\, t \mid \;
                           \rbox\,t \mid \into\,t \mid  \fix  \; x.t
                           \mid l \mid x \mid t_1\,t_2 \mid t_1 + t_2 \mid \adv\,t \\
  & & \mid &  \delay\,t \mid \caseterm{t}{x}{t_1}{x}{t_2} \mid \unbox\,t \mid \progress\,t \mid \promote\,t \mid \out\,t 
\end{array}
\]
  \caption{Syntax.}
  \label{fig:syntax}
\end{figure}

\begin{figure}[tbp]
  \textsc{Well-formed types\quad $\istype{\Theta}{A}$}
  \begin{mathpar}
  \inferrule*
  {\alpha \in \Theta}
  {\istype{\Theta}{\alpha}}
  \and 
  \inferrule*
  {~}
  {\istype{\Theta}{\Unit}}
  \and
  \inferrule*
  {~}
  {\istype{\Theta}{\Nat}}
  \and 
  \inferrule*
  {\istype{\Theta}{A} \\ \istype{\Theta}{B}}
  {\istype{\Theta}{A \times B}}
  \and
  \inferrule*
  {\istype{\Theta}{A} \\ \istype{\Theta}{B}}
  {\istype{\Theta}{A + B}}
  \and
  \inferrule*
  {\istype{\Theta}{A} \\ \istype{\Theta}{B}}
  {\istype{\Theta}{A \to B}}
  \and
  \inferrule*
  {\istype{\Theta}{A}}
  {\istype{\Theta}{\Delay A}}
  \and
  \inferrule*
  {\istype{\Theta}{A}}
  {\istype{\Theta}{\Box A}}
  \and
  \inferrule*
  {\istype{\Theta,\alpha}{A}}
  {\istype{\Theta}{\mu\alpha.A}}
\end{mathpar}
\\\vspace{1em}
\textsc{Well-formed contexts\quad $\wfcxt{\Gamma}$}
\begin{mathpar}
  \inferrule*
  {~}
  {\wfcxt{\emptyset}}
  \and
  \inferrule*
  {\wfcxt{\Gamma}\\\istype{}{A}}
  {\wfcxt{\Gamma,x:A}}
  \and
  \inferrule*
  {\wfcxt{\Gamma} \\ \tokenfree{\Gamma}}
  {\wfcxt{\Gamma,\lock}}
  \and
  \inferrule*
  {\wfcxt{\Gamma} \\ \tickfree{\Gamma} \\ \lock \in \Gamma}
  {\wfcxt{\Gamma,\tick}}
\end{mathpar}
\\\vspace{1em}
\textsc{Stable types \quad $A\; \stable$}
\begin{mathpar}
  \inferrule*
  {~}
  {\Unit \; \stable}
  \and
  \inferrule*
  {~}
  {\Nat \; \stable}
  \and
  \inferrule*
  {~}
  {\Box A \; \stable}
  \and
  \inferrule*
  {A \; \stable \\ B \; \stable}
  {A \times B \; \stable}
  \and
  \inferrule*
  {A \; \stable \\ B \; \stable}
  {A + B \; \stable}
\end{mathpar}

\caption{Context and type formation rules for Simply RaTT.}
\label{fig:types}
\end{figure}

\begin{figure}[tbp]
  \begin{mathpar}
      \inferrule*
      {\wfcxt{\Gamma,x:A,\Gamma'} \\ \tokenfree{\Gamma'}}
      {\hastype{\Gamma,x:A,\Gamma'}{x}{A}}
      \and
      \inferrule*
      {\wfcxt{\Gamma}}
      {\hastype{\Gamma}{\unit}{\Unit}}
      \and
      \inferrule*
      {n \in \nats}
      {\hastype{\Gamma}{\bar{n}}{\Nat}}
      \and
      \inferrule*
      {\hastype{\Gamma}{s}{\Nat} \\ \hastype{\Gamma}{t}{\Nat}}
      {\hastype{\Gamma}{s + t}{\Nat}}
      \and
      \inferrule*
      {\hastype{\Gamma,x:A}{t}{B}\\\tickfree{\Gamma}}
      {\hastype{\Gamma}{\lambda x .t}{A \to B}}
      \and
      \inferrule*
      {\hastype{\Gamma}{t}{A \to B} \\
        \hastype{\Gamma}{t'}{A}}
      {\hastype{\Gamma}{t\,t'}{B}}
      \and
      \inferrule*
      {\hastype{\Gamma}{t}{A} \\
        \hastype{\Gamma}{t'}{B}}
      {\hastype{\Gamma}{\pair{t}{t'}}{A \times B}}
      \and
      \inferrule*
      {\hastype{\Gamma}{t}{A_1 \times A_2} \\ i \in \{1, 2\}}
      {\hastype{\Gamma}{\pi_i\, t}{A_i}}
      \and
      \inferrule*
      {\hastype{\Gamma}{t}{A_i} \\ i \in \{1, 2\}}
      {\hastype{\Gamma}{\interm_i\, t}{A_1 + A_2}}
      \and
      \inferrule*
      {\hastype{\Gamma,x: A_i}{t_i}{B}
        \\ \hastype{\Gamma}{t}{A_1 + A_2}
        \\ i \in \{1,2\}}
      {\hastype{\Gamma}{\caseterm {t}{x}{t_1}{x}{t_2}}{B}}
      \and
      \inferrule*
      {\hastype{\Gamma,\tick}{t}{A}}
      {\hastype{\Gamma}{\delay\,t}{\Delay A}}
      \and
      \inferrule*
      {\hastype{\Gamma}{t}{\Delay A} \\ \wfcxt{\Gamma,\tick,\Gamma'}}
      {\hastype{\Gamma,\tick,\Gamma'}{\adv\,t}{A}}
      \and
      \inferrule*
      {\hastype{\Gamma}{t}{\Box A} \\ \tokenfree{\Gamma'}}
      {\hastype{\Gamma,\lock,\Gamma'}{\unbox\,t}{A}}
      \and
      \inferrule*
      {\hastype{\Gamma,\lock}{t}{A}}
      {\hastype{\Gamma}{\rbox\,t}{\Box A}}
      \and
      \inferrule*
      {\hastype{\Gamma}{t}{A} \\ \wfcxt{\Gamma,\tick,\Gamma'} \\ A \; \stable}
      {\hastype{\Gamma,\tick,\Gamma'}{\progress\,t}{A}}
      \and
      \inferrule*
      {\hastype{\Gamma}{t}{A} \\ \wfcxt{\Gamma,\lock,\Gamma'} \\ A \; \stable}
      {\hastype{\Gamma,\lock,\Gamma'}{\promote\,t}{A}}
      \and
      \inferrule*
      {\hastype{\Gamma}{t}{A[\Delay(\mu\alpha.A)/\alpha]}}
      {\hastype{\Gamma}{\into\,t}{\mu\alpha.A}}
      \and
      \inferrule*
      {\hastype{\Gamma}{t}{\mu\alpha.A}}
      {\hastype{\Gamma}{\out\,t}{A[\Delay(\mu\alpha.A)/\alpha]}}
      \and
      \inferrule*
      {\hastype{\Gamma,\lock,x : \Delay A}{t}{A}}
      {\hastype{\Gamma}{\fix \; x.t}{\Box A}}
    \end{mathpar}
  \caption{Typing rules of Simply RaTT.}
\label{fig:typing}
\end{figure}

\section{Operational Semantics}
\label{sec:oper-semant}

Following the idea of \citet{krishnaswami13frp}, we devise an
operational semantics for Simply RaTT that is free of space leaks
\emph{by construction}. To this end, the operational semantics is
defined in terms of a machine that has access to a store consisting of
up to two separate heaps: A `now' heap $\eta_N$ from which we can
retrieve delayed computations, and a `later' heap $\eta_L$ where we
can store computations that should be performed in the next time
step. Once the machine advances to the next time step, it will delete
the `now' heap $\eta_N$ and the `later' heap $\eta_L$ will become the
new `now' heap. Thus the problem of proving the absence of space leaks
is reduced to the problem of soundness, i.e., that well-typed programs
never get stuck.

\subsection{Term Semantics}
\label{sec:term-semantics}

\begin{figure}
  \begin{mathpar}
  \inferrule*
  {~}
  {\hevalSingle{v}{\sigma}{v}{\sigma}}
  \and
  \inferrule*
  {\hevalSingle {t} {\sigma} {u} {\sigma'}\\
    \hevalSingle {t'} {\sigma'} {u'} {\sigma''}}
  {\hevalSingle {\pair{t}{t'}} {\sigma} {\pair{u}{u'}} {\sigma''}}
  \and
  \inferrule*
  {\hevalSingle {t} {\sigma} {\pair{v_1}{v_2}} {\sigma'} \\ i \in \{1,2\}}
  {\hevalSingle {\pi_i(t)} {\sigma} {v_i} {\sigma'}}
  \and
  \inferrule*
  {\hevalSingle t {\sigma} v {\sigma'} \\ i \in \{1,2\}}
  {\hevalSingle {\interm_i(t)} {\sigma} {\interm_i(v)} {\sigma'}}
  \and
  \inferrule*
  {\hevalSingle {t} {\sigma} {\interm_i(u)} {\sigma'}\\
    \hevalSingle {t_i[v/x]} {\sigma'} {u_i} {\sigma''} \\ i \in \{1,2\}}
  {\hevalSingle {\caseterm{t}{x}{t_1}{x}{t_2}} {\sigma} {u_i} {\sigma''}}
  \and
  \inferrule*
  {\hevalSingle{t}{\sigma}{\lambda x.s}{\sigma'} \\
    \hevalSingle{t'}{\sigma'}{v}{\sigma''}\\
    \hevalSingle {s[v/x]}{\sigma''}{v'}{\sigma'''}}
  {\hevalSingle{t\,t'}{\sigma}{v'}{\sigma'''}}
  \and
  \inferrule*
  {\hevalSingle{t}{\sigma}{\ol m}{\sigma'} \\
    \hevalSingle{t'}{\sigma'}{\ol n}{\sigma''}}
  {\hevalSingle{t + t'}{\sigma}{\ol{m+n}}{\sigma''}}
  \and 
  \inferrule*
  {\sigma \neq \bot \\ l = \allocate{\sigma}}
  {\hevalSingle{\delay\,t}{\sigma}{l}{\sigma,l \mapsto t}}
  \and
  \inferrule*
  {\hevalSingle {t}{\lock\eta_N}{l}{\lock\eta_N'} \\
    \hevalSingle {\eta_N'(l)}{\lock\eta_N'\tick\eta_L} {v}{\sigma'}}
  {\hevalSingle {\adv\,t}{\lock\eta_N\tick\eta_L}{v}{\sigma'}}
  \and
  \inferrule*
  {\hevalSingle {t} {\bot} {v} {\bot} \\ \sigma \neq \bot}
  {\hevalSingle {\promote\,t} {\sigma} {v} {\sigma}}
  \and
  \inferrule*
  {\hevalSingle {t} {\lock\eta_N} {v} {\lock\eta'_N}}
  {\hevalSingle {\progress\,t} {\lock\eta_N\tick\eta_L} {v} {\lock\eta'_N\tick\eta_L}}
  \and
  \inferrule*
  {\hevalSingle{t}{\bot}{\rbox\,t'}{\bot} \\
    \hevalSingle{t'}{\sigma}{v}{\sigma'} \\ \sigma \neq \bot}
  {\hevalSingle{\unbox\,t}{\sigma}{v}{\sigma'}}
  \and 
  \inferrule*
  {\hevalSingle{t}{\sigma}{v}{\sigma'}}
  {\hevalSingle{\into\,t}{\sigma}{\into\,v}{\sigma'}}
  \and
  \inferrule*
  {\hevalSingle{t}{\sigma}{\into\,v}{\sigma'}}
  {\hevalSingle{\out\,t}{\sigma}{v}{\sigma'}}
  \and
  \inferrule*
  {\hevalSingle{t}{\bot}{\fix \; x.t'}{\bot} \\
    \hevalSingle{t'[l/x]}{\sigma,l \mapsto \unbox(\fix \;
      x.t')}{v}{\sigma'} \\ \sigma \neq \bot \\
    l =  \allocate{\sigma}}
  {\hevalSingle{\unbox\,t}{\sigma}{v}{\sigma'}}
\end{mathpar}
  \caption{Big-step operational semantics.}
  \label{fig:machine}
\end{figure}

\begin{figure}
  \begin{mathpar}
    \inferrule*
    {\hevalSingle{t}{\lock\eta\tick}{v :: l}{\lock\eta_N\tick\eta_L}}
    {\state{t}{\eta} \forwards{v} \state{\adv\,l}{\eta_L}}
    \and
    \inferrule*
    {\hevalSingle{t}{\lock\eta,\thel \mapsto v :: \thel \tick\thel
        \mapsto \unit}{v' :: l}{\lock\eta_N\tick\eta_L,\thel
        \mapsto \unit}}
    {\state{t}{\eta} \forward{v}{v'} \state{\adv\,l}{\eta_L}}
  \end{mathpar}
  \caption{Small-step operational semantics for stream unfolding and
    stream processing.}
  \label{fig:stream_machine}
\end{figure}

The operational semantics of terms is presented in
\autoref{fig:machine}. Given a term $t$ together with a store
$\sigma$, we write $\hevalSingle{t}{\sigma}{v}{\sigma'}$ to denote
that the machine evaluates $t$ in the context of $\sigma$ to a value
$v$ and produces an updated store $\sigma'$. Importantly, a store
$\sigma$ can take on three different forms: It may contain no heap,
written $\sigma = \bot$; it may consist of one heap $\eta_L$, written
$\sigma = \lock\eta_L$; or it may consist of two heaps $\eta_N$ and
$\eta_L$, written $\sigma = \lock\eta_N\tick\eta_L$. These different
forms of stores enforce effective restrictions on when the machine is
allowed to store or retrieve delayed computations. If $\sigma = \bot$,
then computations may neither be stored nor retrieved. If
$\sigma = \lock\eta_L$, then computations may be stored in $\eta_L$ to
be retrieved in the next time step. And if
$\sigma = \lock\eta_N\tick\eta_L$, computations may be stored in
$\eta_L$ as well as retrieved from $\eta_N$. Heaps themselves are
simply finite mappings from \emph{heap locations} to terms.

Given a store $\sigma$ that is not $\bot$, i.e., it is either of the
form $\lock\eta_L$ or $\lock\eta_N\tick\eta_L$, the machine can store
delayed computations on the `later' heap $\eta_L$. To this end, we use
the notation $\later{\sigma}$ to refer to $\eta_L$, and given
$l \nin \dom{\eta_L}$, we write $\sigma,l\mapsto t$ for the store
$\lock(\eta_L,l\mapsto t)$ or $\lock\eta_N\tick(\eta_L,l\mapsto t)$,
respectively. In turn, $\eta_L,l\mapsto t$ denotes the heap obtained
by extending $\eta_L$ with a new mapping $l\mapsto t$. To allocate a
fresh heap locations, we assume a function $\allocate{\cdot}$ that
takes a store $\sigma \neq \bot$ and returns a heap location $l$ such
that $l \nin \dom{\later{\sigma}}$. That is, given
$l = \allocate\sigma$, we can form the new store $\sigma,l\mapsto t$
without overwriting any mappings that are present in $\sigma$.

As the notation suggests, there is a close correspondence between the
shape of a context $\Gamma$ and the shape of a store $\sigma$. Terms
typable in an initial judgement (cf.\ \autoref{fig:judgement}) can be
executed safely with a store $\bot$ -- they need not store nor
retrieve delayed computations. Terms typable in a now judgement can be
executed safely with a store $\lock \eta_L$ or
$\lock \eta_N \tick \eta_L$ -- they may store delayed computations in
$\eta_L$, but need not retrieve delayed computations. And finally,
terms typable in a later judgement can be executed safely in a store of
the form $\lock\eta_N\tick\eta_L$ -- they may retrieve delayed
computations from $\eta_N$.

This intuition of the capabilities of the different stores can be
observed directly in the semantics for $\delay$ and $\adv$: For
$\delay\,t$ to evaluate, the machine expects a store that is not
$\bot$, i.e., a store $\lock \eta_L$ or $\lock \eta_N \tick
\eta_L$. Then the machine allocates a fresh heap location $l$ in the
heap $\eta_L$ and stores $t$ in it. This corresponds to the fact that
$\delay\,t$ can only be typed in a now judgement. Conversely,
$\adv\, t$ requires the store to be of the form
$\lock \eta_N \tick \eta_L$ so that $t$ can be evaluated safely with
the store $\lock \eta_N$ to a heap location $l$, which either already
existed in $\eta_N$ or was allocated when evaluating $t$. In either
case, the delayed computation stored at heap location $l$ is retrieved
and executed. The combinator $\progress$ with its typing rule similar
to that of $\adv$, also has similar operational behaviour in terms of
how it interacts with the store.

Fixed points are evaluated when a term $t$ that evaluates to a value
of the form $\fix\,x.t'$ is unboxed. For a general fixed point
combinator, we would expect that $\fix\,x.t'$ unfolds to
$t'[\fix\,x.t'/x]$. In our setting, the types dictate that
$\fix\,x.t'$ should rather unfold to
$t'[\delay(\unbox(\fix\,x.t'))/x]$, because $x$ has type $\Delay A$
and $\fix\,x.t'$ has type $\Box A$. This is close to the behaviour of
our machine (and would in fact be a safe alternative
definition). Instead, however, the machine anticipates that the term
allocates a mapping $l \mapsto \unbox(\fix\,x.t')$ on the store and
evaluates to that heap location $l$. Therefore, the machine evaluates
the fixed point by allocating a mapping $l \mapsto \unbox(\fix\,x.t')$
on the store right away and evaluating $t'[l/x]$ subsequently.

\subsection{Stream Semantics}
\label{sec:stream-semantics}

The careful distinction between a `now' heap $\eta_N$ and a `later'
heap $\eta_L$ is crucial in order to avoid \emph{implicit} space
leaks. After the machine has evaluated a term $t$ to a value $v$ and
produced a store of the form $\lock \eta_N \tick \eta_L$, we can
safely garbage collect the entire heap $\eta_N$ and compute the next
step with the store $\lock\eta_L\tick$. For example, if the original
term $t$ was of type $\Str\Nat$, then its value $v$ will be of the
form $\ol n :: l$, where $n$ is the head of the stream and $l$ is a
heap location that points to the delayed computation that computes the
tail of the steam. The tail of the stream can then be safely computed
by evaluating $\adv\, l$ with the store $\lock \eta_L\tick$, i.e.
with the entire `now' heap $\eta_N$ garbage collected.

This idea of computing streams is made formal in the definition of the
small-step operational semantics $\forwards{\cdot}$ for streams given
in the left half of \autoref{fig:stream_machine}. It starts by
evaluating the term to a value of the form $v :: l$, which
additionally produces the store $\lock\eta_N\tick\eta_L$. Then the
computation can be continued by evaluating $\adv\, l$ in the garbage
collected store $\lock\eta_L\tick$, which in turn produces a value
$v' :: l'$ and a store $\lock\eta'_N\tick\eta'_L$ -- and so on.

Given a closed term $t$ of type $\Box \Str A$, we compute the elements
$v_1,v_2,v_3, \dots$ of the stream defined by $t$ as follows:
\[
  \state{\unbox\,t}{\emptyset} \forwards{v_1} \state{t_1}{\eta_1}
  \forwards{v_2} \state{t_2}{\eta_2}  \forwards{v_3} \dots
\]
where we start with the empty heap $\emptyset$. Each state of the
computation $\state{t_i}{\eta_i}$ consists of a term $t_i$ and its
`now' heap $\eta_i$.

For example, consider the stream
$\hastype{}{\Varid{nats}}{\Box(\Str\Nat)}$ defined in
\autoref{sec:stable-types}. To understand how this stream is executed,
it is helpful to see how the definition of $\Varid{nats}$ desugars to
our core calculus. Namely, $\Varid{nats}$ is defined as the term
$\Varid{from} \promoteapp \ol 0$ where
\[
  \Varid{from} = \fix f.\lambda n. n :: \delay\, (\adv\, f\,
  (\progress\,(n+\ol 1)))
\]
Here we also unfold the definition of $\progressapp$. The first three
steps of executing the $\Varid{nats}$ stream look as follows:
\[
  \begin{array}{rcl}
  \state{\unbox\,\Varid{nats}}{\emptyset}
  &\forwards{\ol 0}& \state{\adv\,l'_1}{l_1 \mapsto
    \unbox\,\Varid{from},\,l'_1\mapsto \adv\,l_1\,(\progress\, \ol 0 +
                     \ol 1)}\\
  &\forwards{\ol 1}& \state{\adv\,l'_2}{l_2 \mapsto
    \unbox\,\Varid{from},\,l'_2\mapsto \adv\,l_2\, (\progress\, \ol 1 + \ol 1)}\\
  &\forwards{\ol 2}& \state{\adv\,l'_3}{l_3 \mapsto
    \unbox\,\Varid{from},\,l'_3\mapsto \adv\,l_3\,(\progress\, \ol 2 + \ol 1)}\\
  &\vdots&
  \end{array}
\]
As expected, the computation produces the consecutive natural
numbers. In each step of the computation, the location $l_i$ stores
the fixed point $\Varid{from}$ that underlies $\Varid{nats}$, whereas
$l'_i$ stores the computation that calls that fixed point with the
current state of the computation, namely the number $\ol i$.

Our main result is that execution of programs by the machine in
\autoref{fig:machine} and \autoref{fig:stream_machine} is
safe. For the stream semantics, this means that we can compute the
stream defined by a term $t$ of type $\Box(\Str A)$ by successive
unfolding ad infinitum as follows:
\[
  \state{\unbox\,t}{\emptyset} \forwards{v_1} \state{t_1}{\eta_1}
  \forwards{v_2} \state{t_2}{\eta_2}  \forwards{v_3} \dots
\]
This intuition is expressed more formally in the following theorem:
\begin{theorem}[productivity]
  \label{thr:productivity}
  Let $A$ be a value type, i.e., a type constructed from
  $\Unit,\Nat, +,\times$ only, and $\hastype{}{t}{\Box(\Str
    A)}$. Given any $n \in \nats$, there is a reduction sequence
  \[
    \state{\unbox\,t}{\emptyset} \forwards{v_1} \state{t_1}{\eta_1}
    \forwards{v_2} \quad \dots \quad  \forwards{v_n}
    \state{t_n}{\eta_n} \quad \text{such that $\hastype{}{v_i}{A}$ for all $1\le i \le n$.}
  \]
\end{theorem}

\subsection{Stream Transducer Semantics}
\label{sec:stre-transd-semant}

More importantly, our language also facilitates stream processing,
that is executing programs of type $\Box (\Str A \to \Str B)$. The
small-step operational semantics $\forward{\cdot}{\cdot}$ for
executing such programs is given on the right half of
\autoref{fig:stream_machine}. So far the store has only been used by
the term semantics to store delayed computations. In addition to that
purpose, the stream transducer semantics uses the store to transfer
the data received from the input steam to the stream transducer. To
this end, we assume an arbitrary but fixed heap location $\thel$,
which the machine uses to successively insert the input stream of type
$\Str A$ as it becomes available. Note that the stream transducer
semantics reserves the heap location $\thel$ in new `later' heap by
storing $\unit$ in it. That means, $\thel$ cannot be allocated by the
machine and is available later when the input becomes available and
needs to be stored in $\thel$.

Given a closed term $t$ of type $\Box (\Str A \to \Str B)$, we can
execute it as follows:
\[
  \state{\unbox\,t\,(\adv\, \thel)}{\emptyset} \forward{v_1}{v'_1}
  \state{t_1}{\eta_1} \forward{v_2}{v'_2} \state{t_2}{\eta_2}
  \forward{v_3}{v'_3} \dots
\]
The machine starts with an empty heap $\emptyset$. In each step
$\state{t_i}{\eta_i} \forward{v_{i+1}}{v'_{i+1}}
\state{t_{i+1}}{\eta_{i+1}}$, the machine starts in a state
$\state{t_i}{\eta_i}$ consisting of a term $t_i$ and heap
$\eta_i$. Then it reads an input $v_{i+1}$ and subsequently produces
the output $v'_{i+1}$ and the next state
$\state{t_{i+1}}{\eta_{i+1}}$.

Let's consider a simple stream transducer to illustrate the workings
of the semantics. The stream transducer $\Varid{sum}$ takes a stream
of numbers and computes at each point in time the sum of all previous
numbers from the input stream. To this end $\Varid{sum}$ uses the
auxiliary function $\Varid{sum'}$ that takes as additional argument
the accumulator of type $\Nat$.
\genSum

To appreciate the workings of the stream transducer semantics, we
desugar the definition of $\Varid{sum'}$ in the surface syntax to our
core calculus. In addition, we also unfold the definition of
$\delayapp$:
\begin{align*}
  \Varid{sum'} &= \fix f. \lambda acc. \lambda s. (acc + \head\, s) ::
  \delay\, (\adv\, (f \progressapp
  (acc + \head\, s))\, (\adv\,(\tail\, s)))
\end{align*}
Let's look at the first three steps of executing the $\Varid{sum}$
stream transducer. To this end, we feed the computation $2$, $11$, and
$5$ as input:
\[
  \begin{array}{rcl}
  &&\state{\unbox\,\Varid{sum}}{\emptyset}\\
  &\forward{\ol 2}{\ol 2}&
  \state{\adv\,l'_1}{l_1 \mapsto \unbox\,\Varid{sum'},\,l'_1\mapsto
    \adv\,(l_1\progressapp (\ol 0 + \head\, (\ol 2 :: \thel)))\,(\adv\, (\tail\,
    (\ol 2 :: \thel)))}\\
  &\forward{\ol{11}}{\ol{13}}&
  \state{\adv\,l'_2}{l_2 \mapsto \unbox\,\Varid{sum'},\,l'_2\mapsto
    \adv\,(l_2\progressapp (\ol 2 + \head\, (\ol{11} :: \thel)))\,(\adv\, (\tail\,
    (\ol{11} :: \thel)))}\\
  &\forward{\ol 5}{\ol{18}}&
  \state{\adv\,l'_3}{l_3 \mapsto \unbox\,\Varid{sum'},\,l'_3\mapsto
    \adv\,(l_3\progressapp (\ol{13} + \head\, (\ol 5 :: \thel)))\,(\adv\, (\tail\,
                           (\ol 5 :: \thel)))}\\
    &\vdots&
  \end{array}
\]
As expected, we receive $2$, $13$ ($= 2 + 11$), and $18$
($= 2 + 11 + 5$) as result. Moreover, in each step of the computation
the location $l_i$ stores the fixed point $\Varid{sum'}$ that underlies
the definition of $\Varid{sum}$, whereas $l'_i$ stores the computation
that calls that fixed point with the new accumulator value ($0 + 2$,
$2 + 11$, and $13 + 5$, respectively) and the tail of the input
stream.

Corresponding to the productivity property from the previous section,
we prove the following causality property that states that the stream
transducer semantics never gets stuck. To characterise the causality
property, the theorem constructs a family of sets $T_k(A,B)$ which
consists of states $\state{t}{\eta}$ on which the stream transducer
machine can run for $k$ more time steps.
\begin{theorem}[causality]
  \label{thr:causality}
  Given any value types $A$ and $B$, there is a family of sets
  $T_k(A,B)$ such that the following holds for all $k \in \nats$:
  \begin{enumerate}[(i)] \item If
    $\hastype{}{t}{\Box(\Str A \to \Str B)}$ then
    $\state{\unbox\, t\,(\adv\,\thel)}{\emptyset} \in T_k(A,B)$.
  \item If $\pair{t}{\eta} \in T_{k+1}(A,B)$ and $\hastype{}{v}{A}$
    then there are $t', \eta'$, and $\hastype{}{v'}{B}$ such that
    \[
    \state{t}{\eta} \forward{v}{v'} \state{t'}{\eta'} \text{ and }
    \state{t'}{\eta'} \in T_k(A,B).
    \]
  \end{enumerate}
\end{theorem}


That is, any term $t$ of type $\Box(\Str A\to \Str B)$ defines a
causal stream function which is effectively computed by the machine:
\[
  \state{\unbox\,t}{\emptyset} \forward{v_1}{v'_1} \state{t_1}{\eta_1}
  \forward{v_2}{v'_2} \state{t_2}{\eta_2}  \forward{v_3}{v'_3} \dots
\]

Note that the stream transducer semantics also extends to stream
transducers with multiple streams as inputs. This can be achieved by a
combinator $\Varid{split}$ of type
$\Box (\Str{A \times B} \to \Str A \times \Str B)$. Similarly, the
semantics also extends to transducers that take an event as input by
virtue of a combinator $\Varid{first}$ of type
$\Box (\Str{1 + A} \to \Ev A)$. Conversely, transducers producing
events instead of streams can be executed using a combinator of type
$\Box (\Ev A \to \Str{1 + A})$.

We give the proof of \autoref{thr:productivity} and
\autoref{thr:causality} in \autoref{sec:metatheory}. Both results
follow from a more general result for the machine, which is formulated
using a Kripke logical relation.

\subsection{Counterexamples}
\label{sec:counterexamples}

To conclude this section we review some programs that are rejected by
our type system and illustrate their operational behaviour.

\paragraph{Unbox under delay.}
Recall the alternative definition of the stream of consecutive natural
numbers $\Varid{leakyNats}$ that uses the $\Varid{map}$ combinator. First
consider the definition of $\Varid{leakyNats}$ in our core calculus:
\[
  \Varid{leakyNats} = \fix\,s. \ol 0 :: \delay\,(\unbox\,(map\,(\rbox(\lambda
  x. x + \ol 1)))) \delayapp s
\]

Let's contrast the execution of $\Varid{nats}$ that we have seen in
\autoref{sec:stream-semantics} with the execution of $\Varid{leakyNats}$:
\[
      \arraycolsep=1.4pt
  \begin{array}{rc@{\quad}l}
  &&\state{\unbox\,\Varid{leakyNats}}{\emptyset}\\
    &\forwards{\ol 0}& \state{\adv\,l'_1}{
                       \begin{array}[c]{ll}
                       l_1 \mapsto
    \unbox\,\Varid{leakyNats},&l'_1\mapsto \unbox\,\Varid{map}\, (\rbox\,\lambda
  x. x + \ol 1)\, (\adv\, l_1)
                       \end{array}
}\\[1.5em]
    &\forwards{\ol 1}& \state{\adv\,l^{3}_2}{
                       \begin{array}[c]{ll}
                         l^0_2 \mapsto \unbox\,\Varid{leakyNats}, &l^1_2
                         \mapsto \unbox\,\Varid{map}\, (\rbox\,\lambda
                         x. x + \ol 1)\, (\adv\, l^0_2),\\ l^{2}_2
                         \mapsto \unbox\,\Varid{step},&l^{3}_2 \mapsto
                         \adv\, l^{2}_2\,(\adv\,(\tail\, (\ol 0 :: l^1_2)))
                       \end{array}}\\[2em]
    &\forwards{\ol 2}& \state{\adv\,l^{5}_3}{
                       \begin{array}[c]{ll}
                         l^0_3 \mapsto \unbox\,\Varid{leakyNats},&l^1_3
                         \mapsto \unbox\,\Varid{map}\, (\rbox\,\lambda
                         x. x + \ol 1)\, (\adv\, l^0_3),\\ l^{2}_3
                         \mapsto \unbox\,\Varid{step},&l^{3}_3 \mapsto
                         \adv\, l^{2}_3\,(\adv\,(\tail\, (\ol 0 ::
                         l^1_3)))\\
                         l^{4}_3
                         \mapsto \unbox\,\Varid{step},&l^{5}_3 \mapsto
                         \adv\, l^{4}_3\,(\adv\,(\tail\, (\ol 1 :: l^3_3)))
                       \end{array}}\\
  &\vdots&
  \end{array}
\]
where
$\Varid{step} = \fix\, f. \lambda\, s. \unbox\, (\rbox\, \lambda\,
n. n + \ol 1)\, (\head\, s) :: (f \delayapp \tail\, s)$.


While our type system rejects the term $\Varid{leakyNats}$, a
corresponding term is typable in \citeauthor{krishnaswami13frp}'s
calculus~\citep{krishnaswami13frp} and manifests the same memory
allocation behaviour as $\Varid{leakyNats}$ in our machine.

\paragraph{Lambda abstraction under delay.}

Recall the definition of the stream $\Varid{leaky}$ from
\autoref{sec:function-types}. It introduces a lambda abstraction in a
later judgement and is therefore rejected by our type system.
\[
  \begin{array}{rcl}
  &&\state{\unbox\,\Varid{leaky}}{\emptyset}\\[.5em]
    &\forwards{\sym{true}}& \state{\adv\,l'_1}{
                            \begin{aligned}[c]
                              &l_1 \mapsto
                              \unbox\,\Varid{leaky'},\\
                              &l'_1\mapsto \adv\, (\mathbf{if}\,
                              (\lambda\, x.\sym{true})\,
                              \unit\,\mathbf{then}\,
                              l_1\,\mathbf{else}\, l_1)\,
                              (\lambda\,x. \head\, (\adv\, l_1\,
                              \lambda\,y. \sym{true}))
                            \end{aligned}
} \\[1.5em]
    &\forwards{\sym{true}}& \state{\adv\,l'_2}{
                            \begin{aligned}[c]
                              &l_2 \mapsto
                              \unbox\,\Varid{leaky'},\\
                              &l'_2\mapsto \adv\, (\mathbf{if}\,
                              (\lambda\, x.\head\,(\adv\, l_1\,\lambda\,y.\sym{true}))\,\unit\,\mathbf{then}\,
                              l_2\,\mathbf{else}\, l_2)\\
                              &\phantom{l'_2\mapsto \adv\, }(\lambda\,x. \head\, (\adv\, l_2\,
                              \lambda\,y. \sym{true}))
                            \end{aligned}
                                } \\
    &\centernot{\forwards{\cdot}}&
  \end{array}
\]
Note that the term
$(\lambda\, x.\head\,(\adv\, l_1\,\lambda\,y.\sym{true}))$ from the
heap after the first step is a value and thus appears unevaluated also
in the heap after the second step. However, this term contains a
reference to the heap location $l_1$, which has been garbage collected
completing the second step. The machine thus ends up in a stuck state
when it tries to dereference the garbage collected heap location $l_1$
during the third step.

\section{Generic FRP Library}
\label{sec:stream_library}
This section gives a number of higher-order FRP combinators in Simply
RaTT, reminiscent of those found in libraries such as Yampa
\cite{nilsson2002}.  These can be used for programming with streams
and events.

Perhaps the simplest example of a stream function is the constant
stream over some element.  Since we need to output this element in
each time step in the future, we require it to come from a stable
type. Thus the argument is of type $\Box\, A$.
\genConstGeneral
We can now recreate the $\Varid{zeros}$ stream presented above as:
\genConstZeros
Another simple way to generate a stream is to iterate a function
$f : A \to A$ over some initial input, such that the output stream
will be $(a,f a,f (f a),\ldots)$.  Since we will keep using the
function at every time step, it needs to be stable, i.e.,
$f: \Box(A \to A)$. Moreover, since we will keep a state with the
current value of type $A$, that type $A$ must be stable as well. We
adopt a syntax like the one used in Haskell for type classes, to
denote the additional requirement that a type is stable:
\genIter
With this, we can define the stream of natural numbers:
\genIterNats
We may also define a more general $\Varid{iter}$ where $A$ need
not be stable:
\genIterGeneral

Given some stream, a standard operation in an FRP setting is to
\emph{filter} it according to some predicate.  This behaviour is easy
to implement in Simply RaTT, but because productivity forces us to
output a value at each time step, if we take as input $\Str{A}$, we
will need to output $\Str{\Maybe{A}}$, where $\Maybe{A}$ is a
shorthand for $1 + A$. Accordingly, we use the notation $\nothing$ and
$\just\,t$ to denote $\interm_1\, \unit$ and $\interm_2\,t$,
respectively.
%
%
\genFilterMap
To go from $\Str{\Maybe{A}}$ and back to $\Str{A}$, we can use the
$\Varid{fromMaybe}$ function that replaces each missing value with a
default value:
\genFromMaybe


Given two streams, we can construct the product stream by simply ``zipping'' the two streams together.
It is often easier to construct the more general version where a function is applied to each pair of inputs
  \genZipWith
The regular zip function is then defined as
  \genZip

Many applications require the ability to dynamically change the dataflow graph, e.g., when opening and closing 
windows in a GUI. Such behaviour can be implemented using \emph{switches}, such as the following, which
%
%
given an initial stream and a stream event, outputs a stream following the initial stream until it receives a new one on
its second argument
%
  \genSwitch

As we have described above, we may define streams that require a state,
but the state must be defined explicitly.  An example is the
\emph{scan} function that given a binary operator and an initial
state, will output the stream of successive application of the binary
operator on the input stream
\genScan
We can now redefine the $\Varid{sum}$ function  from
\autoref{sec:stre-transd-semant} as follows:
\genSumScan

In general, we can encode any computable stream in our language by
virtue of the following \emph{unfolding} combinator:
\genUnfold

To further showcase programming with events, we define the function
$\Varid{await}$, which listens for two events and produces a pair
event that triggers after both events have arrived. As with
$\Varid{scan}$, we need a state to keep the value of the first arriving event
while waiting for the second one. This behaviour is implemented by two helper functions,
which differ only in which element of the pair is given, and we only show one:
\genAwaitA
We can now define $\Varid{await}$ as
\genAwait
The requirement that $A$ and $B$ be stable is crucial since we need
to keep the first arriving event until the second occurs. 

A second example using events is the \emph{accumulator} combinator.
Given a value and a stream of events carrying functions, every time an
event is received, the function is applied to the value and output
as an event:
\genAccum
The $\Varid{accum}$ function uses the helper function below that
takes a single event carrying a function and produces an event that
applies the function to a given value:
\genEventApp

\section{Simulating Lustre}
\label{sec:lustre}

Lustre is a synchronous dataflow language. Programs describe dataflow graphs and evaluation
proceeds in steps, reading input signals and producing output signals. Each signal is associated
with a clock, which is always a sub-clock of the global clock, and as such can be described as
a sequence of Booleans describing when the clock ticks. A pair of a
clock and a signal that produces an output
whenever the clock ticks is called a \emph{flow}.
%

In Simply RaTT clocks can be encoded as Boolean streams and flows as streams of maybe values
\[
  \sym{Clock} = \Str{\sym{Bool}}\hspace{2cm}
  \sym{Flow}(A) = \Str{\Maybe{A}}
\]
With these encodings, we now show how to encode some of the basic constructions of Lustre. 

The clock associated to a flow ticks whenever the stream produces a value in $A$. For example,
the 
\emph{basic clock} of the system is the fastest possible clock 
and the clock $\Varid{never}$ is the clock that never ticks. These can be defined as
\begin{center}
  \vspace{-1em}
  \begin{minipage}[t]{0.45\linewidth}
    \lustreBasicClock
  \end{minipage}
  \begin{minipage}[t]{0.45\linewidth}
    \lustreNever
  \end{minipage}
\end{center}
In Simply RaTT, a cycle of the program corresponds to a single stream
(transducer) unfolding.

Given a clock, we can slow it down to tick only at certain intervals.
We define here a function that given a clock, slows it down to only tick every $n$th tick.
Since we need to carry a state (how often to tick and what step we are at) we define first a helper function:
  \lustreEveryNthAux
We can now define the actual function by giving the helper function an initial state:
  \lustreEveryNth

Given a flow and a clock, we can restrict the flow to that clock.
If the clock is faster than the ``internal'' clock of the flow, 
this will not change the flow.
  \lustreWhen

Recursive flows are implemented in Lustre using \texttt{pre} (previous) and \texttt{->} (followed by).
We \emph{could} implement these directly in Simply RaTT, but in many Lustre programs \texttt{pre} and \texttt{->} are used in a pattern that is very natural to Simply RaTT programs: \texttt{pre} is used to keep track of some state (that we may update) and \texttt{->} provides the initial state.
As an example, consider the Lustre node that computes the flow of natural numbers:
\begin{center}
  \texttt{n = 0 -> pre(n) + 1}
\end{center}
The equivalent Simply RaTT program is 
  \lustreNats
which is similarly composed of an initial state and then the actual computation, which may refer to the previous value. 

As another example, consider the following Lustre node which takes a Boolean flow \texttt{b} as input:
\begin{center}
  \texttt{edge = false -> (b and not pre(b))}
\end{center}
This output flow is true when it detects a ``rising edge'' in its input flow, i.e., when the input goes from false to true. 
This is translated in a similar way to a helper function that does the computation:
  \lustreEdgeAux
and then a function giving the initial state:
  \lustreEdge

Since a flow is restricted to its internal clock, it may not produce anything at many ticks of the basic clock.
To alleviate this, Lustre provides the \texttt{current} operator, which given a flow on a clock slower than the basic clock, fills the holes in the flow with whatever the latest values was. The equivalent Simply RaTT program is:
 \lustreCurrent 

 Our last example is an implementation of \texttt{counter} from the Lustre V6 manual \cite{lustreManual}.
The \texttt{counter} program takes as input an initial value and a constant that determines how much to increment in each step.
Its current state is stored in $\Varid{pre}$.
It then listens to two flows, an ``increment'' flow and a ``reset'' flow.
If the counter receives true on the increment flow, it increments the counter by the increment constant.
If it receives true on the reset flow, it resets the counter to the initial value and otherwise it will continue in the same state.
  \lustreCounter

\section{Metatheory}
\label{sec:metatheory}

\begin{figure}
  \begin{minipage}[b]{1.0\linewidth}
    \begin{align*}
      \vinterp{\Nat}{\sigma}{\ol{H}} &= \setcom{\overline{n}}{n \in \mathbb{N} } \\
      \vinterp{1}{\sigma}{\ol{H}} &= \{ \unit \} \\
      \vinterp{A \times B}{\sigma}{\ol{H}} &= \setcom{ \pair{v_1}{v_2}}{v_1 \in \vinterp{A}{\sigma}{\ol{H}} \land v_2 \in \vinterp{B}{\sigma}{\ol{H}} } \\
      \vinterp{A + B}{\sigma}{\ol{H}} &= \setcom{\interm_1 \, v}{v \in \vinterp{A}{\sigma}{\ol{H}} \} \cup \{\interm_2 \, v \mid v \in \vinterp{B}{\sigma}{\ol{H}}} \\
      \vinterp{A \to B}{\sigma}{\ol{H}} &= \setcom{ \lambda
        x.t}{\forall \world{\sigma'}{\ol{H}'} \heapge \world{\trim{\sigma}}{\ol{H}}.\forall w \in \vinterp{A}{\sigma'}{\ol{H}'}.t[w/x] \in \einterp{B}{\sigma'}{\ol{H}'} } \\
      \vinterp{\Box A}{\sigma}{\ol{H}} &= \setcom{ v}{\forall \ol{H}' \suffix \ol{H}. \unbox(v) \in
        \einterp{A}{\lock}{\ol{H}'} } \\
      \vinterp{\Delay  A}{\sigma}{()} &= \setcom{l}{l \text{ is any heap location}} \\
      \vinterp{\Delay  A}{\sigma}{H;\ol{H}} &=
      \setcom{ l}{\forall \eta \in H. \sigma(l) \in \einterp{A}{\trim{\sigma}\tick\eta}{\ol{H}} } \\
      \vinterp{\mu\alpha.A}{\sigma}{\ol{H}} &= \setcom{ \into(v)}{ v \in
        \vinterp{A[\Delay(\mu\alpha.A)/\alpha]}{\sigma}{\ol{H}} }
      \\[.7em]
      \einterp{A}{\sigma}{\ol{H}} &=
      \setcom{t}{\forall \world{\sigma'}{\ol{H}'}  \tickge
        \world{\sigma}{\ol{H}}. \exists
        \sigma'', v . \hevalSingle{t}{\sigma'}{v}{\sigma''} \land v \in
        \vinterp{A}{\sigma''}{\ol{H}'}}\\[0.7em]
      \cinterp{\cdot}{\bot}{\ol{H}} &= \{\star\}\\ 
      \cinterp{\Gamma,x:A}{\sigma}{\ol{H}} &=
      \setcom{\gamma[x\mapsto v]}{\gamma \in \cinterp{\Gamma}{\sigma}{\ol{H}}, v
        \in\vinterp{A}{\sigma}{\ol{H}}} \\
      \cinterp{\Gamma,\tick}{\lock\eta_N\tick\eta_L}{\ol{H}} &= \cinterp{\Gamma}{\lock\eta_N}{\set{\eta_L},\ol{H}} \\
      \cinterp{\Gamma,\lock}{\sigma}{\ol{H}} &=
      \bigcup_{\ol{H} \suffix \ol{H}'}\cinterp{\Gamma}{\bot}{\ol{H}'}
      \hspace{1cm}\text{if } \sigma \neq \bot
    \end{align*}
  \end{minipage}
  \hspace{-0.3\linewidth}
  \begin{minipage}[b]{0.25\linewidth}
    \begin{align*}
      \intertext{\textsc{Garbage Collection:}}
      \trim{\bot} &= \bot\\ \trim{\lock\eta_L} &= \lock\eta_L\\
      \trim{\lock\eta_N\tick\eta_L} &= \lock\eta_L
    \end{align*}
    \vspace{0.4cm}
  \end{minipage}
  \caption{Logical Relation.}
  \label{fig:logical_relation}
\end{figure}

Since the operational semantics rules out space leaks by
construction, it only remains to be shown that the type system is
sound, i.e., well-typed terms never get stuck.

To this end, we devise a Kripke logical relation. Essentially, such a
logical relation is a family $\sem{A}_w$ of sets of closed terms that
satisfy the desired soundness property. This family of sets is indexed
by $w$ drawn from a suitable set of `worlds' and is defined
inductively on the structure of the type $A$ and $w$. Then the proof
of soundness is reduced to a proof that $\hastype{}{t}{A}$ implies
$t \in \sem{A}_w$ for all possible worlds.

\subsection{Worlds}
\label{sec:worlds}

To a first approximation, the worlds in our logical relation contain
two components: a store $\sigma$ and a number $n$, written
$\sem{A}^{n}_\sigma$. The number index $n$ allows us to define the
logical relation for recursive types via step-indexing
\citep{appel01indexed}. Concretely, this is achieved by defining
$\sem{\Delay A}^{n+1}_\sigma$ in terms of $\sem{A}^{n}_\sigma$. Since
unfolding recursive types $\mu \alpha. A$ to
$A[\Delay\mu \alpha. A/\alpha]$ introduces a $\Delay$ modality, we
thus achieve that the step index $n$ decreases for recursive types. In
essence, this means that terms in the logical relation
$\sem{A}^{n}_\sigma$ can be executed safely for the next $n$ time
steps starting with the store $\sigma$. Ultimately, the index $\sigma$
enables us to prove that the garbage collection performed by the
stream and stream transducer semantics (cf.\
\autoref{fig:stream_machine}) is sound.

While this setup would be sufficient to prove soundness for the stream
semantics (\autoref{thr:productivity}), it is not enough for the
soundness of the stream transducer semantics
(\autoref{thr:causality}): To characterise our soundness property it
is not enough to require that a term can be executed $n$ more time
steps. We also need to know what the input to a stream transducer of
type $\Str A \to \Str B$ looks like, namely a stream where the first
$n$ elements are values of type $A$. We achieve this by describing
what the heaps should look like in the next $n$ time
steps. Concretely, we assume a finite sequence $(H_1;\dots;H_n)$,
where each $H_i$ is the set of heaps that we could potentially
encounter $i$ time steps into the future.

To summarise, the worlds in our logical relation consist of a store
$\sigma$ and a finite sequence $(H_1;\dots;H_n)$ of sets of
heaps. Instead of, $(H_1;\dots;H_n)$ we also write $\ol H$, and we use
the notation $\world{\sigma}{\ol H}$ to refer to the world consisting
of the store $\sigma$ and the sequence $\ol H$. Intuitively, $\sigma$
is the store for which the term in the logical relation can be safely
executed, whereas each $H_i$ in $\ol H$ contains all heaps for which
the term can be safely executed after $i$ time steps have passed.

A crucial ingredient of a Kripke logical relation is a preorder
$\lesssim$ on the set of worlds such that the logical relation is
closed under that preorder in the sense that $w \lesssim w'$ implies
$\sem{A}_w \subseteq \sem{A}_{w'}$. To this end, we use a partial
order $\heaple$ on heaps, which is the standard partial order on
partial maps, i.e., $\eta \heaple \eta'$ iff $\eta(l) = \eta'(l)$ for
all $l \in \dom \eta$. Moreover, we extend this order to stores in two
different ways, resulting in the two orders $\heaple$ and $\tickle$:
\begin{mathpar}
  \inferrule*
  {~}
  {\bot \heaple \bot}
  \and
  \inferrule*
  {\eta \heaple \eta'}
  {\lock \eta \heaple \lock\eta'}
  \and
  \inferrule*
  {\eta_N \heaple \eta'_N\\\eta_L \heaple \eta'_L}
  {\lock\eta_N \tick \eta_L \heaple \lock\eta'_N\tick\eta'_L}
  \and
  \inferrule*
  {\sigma \heaple \sigma'}
  {\sigma \tickle \sigma'}
  \and
  \inferrule*
  {\eta \heaple \eta'}
  {\lock\eta \tickle \lock\eta''\tick\eta'}
\end{mathpar}
That is, the heap order $\heaple$ is lifted to stores pointwise,
whereas $\tickle$ extends $\heaple$ by defining
$\lock \eta_L \tickle \lock \eta_N\tick \eta_L$. The more general order
$\tickle$ is used in the logical relation, whereas the more
restrictive $\heaple$ is needed to characterise the following property
of the operational semantics:
\begin{lemma} 
  \label{lem:machine_monotone}
  Given any term $t$, value $v$, and pair of stores $\sigma,\sigma'$
  such that $\hevalSingle{t}{\sigma}{v}{\sigma'}$, then
  $\sigma \heaple \sigma'$.
\end{lemma}

We also extend the heap order $\heaple$ to sets of heaps and sequences
of sets of heaps:
\[
\begin{array}{rclcl}
  H &\heaple& H' &\iff& \forall \eta' \in H'.\exists \eta \in
  H. \eta \heaple \eta'\\
  (H_1;H_2;\dots;H_n) &\heaple& (H'_1;H'_2;\dots;H'_n) &\iff&
  \forall 1 \le i\le n. H_i \heaple H'_i
\end{array}
\]
That is, if $H \heaple H'$, then for every heap in $H'$ there is a
smaller one in $H$, and this ordering is lifted pointwise to finite
sequences.

Finally, we combine these orderings to worlds pointwise as well
\[
  \begin{array}{rllclcl}
    \world{\sigma}{\ol{H}} &\heaple& \world{\sigma'}{\ol{H}'} &\iff&
                                                                     \sigma \heaple \sigma' &\land& \ol{H} \heaple \ol{H}'\\
    \world{\sigma}{\ol{H}} &\tickle& \world{\sigma'}{\ol{H}'} &\iff& \sigma \tickle \sigma' &\land& \ol{H} \heaple \ol{H}'
  \end{array}
\]

Note that whenever $\ol{H} \heaple \ol{H}'$, then $\ol{H}$ and
$\ol{H}'$ are of the same length. That is, both $\ol{H}$ and $\ol{H}'$
describe the same number of future time steps. In order to describe a
possible future world, i.e., after some time steps have passed, we use
suffix ordering $\suffix$ on sequences. We say that $\ol H$ is a
suffix of $\ol H'$, written $\ol H \suffix \ol H'$ iff there is a
sequence $\ol H''$ such that $H'$ is equal to the concatenation of
$\ol H''$ and $\ol H$, written $\ol H'';\ol H$. Thus, a suffix
$\ol H \suffix \ol H'$ describes a future state where the prefix
$\ol H''$ has already been consumed.

\subsection{Logical Relation}
\label{sec:logical-relation}

Our logical relation consists of two parts: A \emph{value relation}
$\vinterp{A}{w}{}$ that contains all values that semantically inhabit
type $A$ at the world $w$, and a corresponding \emph{term relation}
$\einterp{A}{w}{}$ containing terms. Given a world
$\world{\sigma}{\ol H}$ we write $\vinterp{A}{\sigma}{\ol H}$ and
$\einterp{A}{\sigma}{\ol H}$ instead of
$\vinterp{A}{\world{\sigma}{\ol H}}{}$ and
$\einterp{A}{\world{\sigma}{\ol H}}{}$, respectively. The two
relations are defined by mutual induction in
\autoref{fig:logical_relation}. More precisely, the two relations
are defined by well-founded induction by the lexicographic ordering on
the triple $(\wlen{\ol{H}},\tsize{A},e)$, where $\wlen{\ol{H}}$ is the
length of $\ol{H}$, $\tsize{A}$ is the size of $A$ defined below, and
$e = 1$ for the term relation and $e = 0$ for the value
relation.
\begin{align*}
  \tsize{\alpha} &= \tsize{\Delay A} = \tsize{\Nat} = \tsize{1} = 1\\
  \tsize{A \times B} &= \tsize{A + B} = \tsize{A \to B} = 1 +
  \tsize{A} + \tsize{B}\\
  \tsize{\Box A} &= \tsize{\mu \alpha. A} =  1 +
  \tsize{A}
\end{align*}
We define the size of $\Delay A$ to be the same as $\alpha$. Thus
$A[\Delay\mu \alpha. A/\alpha]$ is strictly smaller than
$\mu \alpha. A$. This justifies the well-foundedness for recursive
types. For types $\Delay A$, the well-foundedness of the definition
can be observed by the fact that $\ol{H}$ is strictly shorter than
$H;\ol{H}$, which is a shorthand notation for the sequence
$(H;H_1;\dots;H_n)$ where $\ol{H} = (H_1;\dots;H_n)$.

The definition of the value relation for $\unit, \Nat, \times$, and
$+$ is standard. The definition of $\vinterp{\Box A}{\sigma}{\ol{H}}$
expresses the fact that all its inhabitants can be evaluated safely at
any time in the future. To express this, we use suffix ordering
$\suffix$. A value in $\vinterp{\Box A}{\sigma}{\ol{H}}$ may be
unboxed and subsequently evaluated at any time in the future, i.e., in
the context of any suffix of $\ol{H}$.

The value relation for types $\Delay A$ encapsulates the soundness of
garbage collection. The set $\vinterp{\Delay A}{\sigma}{H;\ol{H}}$
contains all heap locations that point to terms that can be executed
safely in the next time step. The notation $\sigma(l)$ is a shorthand
for $\eta_L(l)$ given that $\sigma = \lock \eta_L$ or
$\sigma = \lock\eta_N\tick\eta_L$. Hence, we look up the location $l$
in the `later' heap of $\sigma$ and require that the term that we find
can be executed with the store obtained from $\sigma$ by first garbage
collecting the `now' heap (if present) and extending it with any
future heap drawn from $H$.

Garbage collection is also crucial in the definition of
$\vinterp{A \to B}{\sigma}{\ol{H}}$, which only contains lambda
abstractions that can be applied in a garbage collected store. This
reflects the restriction of the typing rule for lambda abstraction,
which requires the context $\Gamma$ to be tick-free. The
$\Varid{leaky}$ example in \autoref{sec:counterexamples} illustrates
the necessity of this restriction. Semantically, this implies the
following essential property of values:
\begin{lemma}
  \label{lem:vinterpTrim}
  For all $A, \sigma, \ol{H}$, we have that
  $\vinterp{A}{\sigma}{\ol{H}} \subseteq
  \vinterp{A}{\trim{\sigma}}{\ol{H}}$.
\end{lemma}
\noindent
That is, after evaluating a term to a value, we can safely garbage
collect the `now' heap.



Finally, we obtain the soundness of the language by the following
fundamental property of the logical relation
$\einterp{A}{\sigma}{\ol{H}}$.
\begin{theorem}[Fundamental Property]
  \label{thr:lrl}
  Given $\hastype{\Gamma}{t}{A}$, and
  $\gamma \in \cinterp{\Gamma}{\sigma}{\ol{H}}$, then
  $t\gamma \in \einterp{A}{\sigma}{\ol{H}}$.
\end{theorem}
The theorem is proved by a lengthy but entirely standard induction on
the typing relation $\hastype{\Gamma}{t}{A}$. Two crucial ingredients
to the proof are that all logical relations are closed under the
ordering $\tickle$ on worlds, and that
$\cinterp{\Gamma}{\sigma}{\ol{H}}$ captures the correspondence between
the tokens occurring in $\Gamma$ and $\sigma$, namely they have the
same number of locks and $\sigma$ may not have fewer ticks than
$\Gamma$.

\subsection{Soundness of Stream and Stream Transducer Semantics}
\label{sec:prod-caus}

We conclude this section by demonstrating how we can use the
fundamental property of our logical relation for proving the
soundness of the abstract machines for evaluating streams
(\autoref{thr:productivity}) and stream transducers
(\autoref{thr:causality}), which amounts to proving productivity and
causality of the calculus.

First, we observe that the operational semantics is deterministic:
\begin{proposition}[deterministic machine]~
\label{prop:determ}  
  \begin{enumerate}
  \item If $\hevalSingle{t}{\sigma}{v_1}{\sigma_1}$ and
    $\hevalSingle{t}{\sigma}{v_2}{\sigma_2}$, then $v_1 = v_2$ and
    $\sigma_1 = \sigma_2$.
  \item If $\state{t}{\eta} \forwards{v_1} \state{t_1}{\eta_1}$ and
    $\state{t}{\eta} \forwards{v_2} \state{t_2}{\eta_2}$, then
    $v_1=v_2$, $t_1 = t_2$, and $\eta_1 = \eta_2$.
  \item If $\state{t}{\eta} \forward{v}{v_1} \state{t_1}{\eta_1}$ and
    $\state{t}{\eta} \forward{v}{v_2} \state{t_2}{\eta_2}$, then
    $v_1=v_2$, $t_1 = t_2$, and $\eta_1 = \eta_2$.
  \end{enumerate}
\end{proposition}

Before we can prove \autoref{thr:productivity}, we need the following
property of value types, i.e., types constructed from
$\Unit,\Nat, +,\times$
\begin{lemma}
  \label{lem:vinterp_value}
  Let $A$ be a value type and $\world{\sigma}{\ol H}$ a world.
  \begin{enumerate}[(i)]
  \item For all values $v$, we have that
    $v\in \vinterp{A}{\sigma}{\ol H}$ iff $\hastype{}{v}{A}$.
    \label{item:vinterp_value1}
  \item $\vinterp{A}{\sigma}{\ol H}$ is non-empty.
    \label{item:vinterp_value2}
  \end{enumerate}
\end{lemma}
\begin{proof}
  By a straightforward induction on $A$.
\end{proof}

For the proof of \autoref{thr:productivity}, we construct for each
type $A$ the following family of sets $S_k(A)$, which intuitively
contains all states on which the stream semantics can run for $k$ more
steps:
\[
  S_k(A) = \setcom{\state{t}{\eta}}{t \in \einterp{\Str A}{\lock
      \eta\tick}{\set{\emptyset}^k}}
\]
where $\set{\emptyset}$ is the singleton set containing the empty
heap, and $\set{\emptyset}^k$ is the sequence containing $k$ copies of
$\set{\emptyset}$. \autoref{thr:productivity} follows from the
following lemma and the fact that the operational semantics is
deterministic:
\begin{lemma}[productivity]
  Given any value type $A$, the following holds for all $k \in \nats$:
  \begin{enumerate}[(i)]
  \item If
    $\hastype{}{t}{\Box(\Str A)}$ then
    $\state{\unbox\, t}{\emptyset} \in S_k(A)$.
  \item If $\pair{t}{\eta} \in S_{k+1}(A)$ then there are $t', \eta'$,
    and $\hastype{}{v}{A}$ such that
    \[
      \state{t}{\eta} \forwards{v} \state{t'}{\eta'} \text{ and }
      \state{t'}{\eta'} \in S_k(A).
    \]
  \end{enumerate}
\end{lemma}
\begin{proof}~
  \begin{enumerate}[(i)]
  \item $\hastype{}{t}{\Box(\Str A)}$ implies
    $\hastype{\lock}{\unbox\,t}{\Str A}$ which by \autoref{thr:lrl},
    implies that
    $\unbox\,t \in \einterp{\Str
      A}{\lock}{\set{\emptyset}^k}$ and thus also $\unbox\,t \in \einterp{\Str
      A}{\lock\tick}{\set{\emptyset}^k}$. Hence
    $\state{\unbox\,t}{\emptyset} \in S_k(A)$.
  \item Let $\pair{t}{\eta} \in S_{k+1}(A)$. Then
    $t \in \einterp{\Str A}{\lock \eta\tick}{\set{\emptyset}^{k+1}}$,
    which means that $\hevalSingle{t}{\lock \eta\tick}{w}{\sigma}$ and
    $w \in \vinterp{\Str A}{\sigma}{\set{\emptyset}^{k+1}}$. Hence,
    $w = v :: l$ with
    $v \in \vinterp{A}{\sigma}{\set{\emptyset}^{k+1}}$ and
    $l \in \vinterp{\Delay \Str
      A}{\sigma}{\set{\emptyset}^{k+1}}$. Moreover, by
    \autoref{lem:vinterp_value} $\hastype{}{v}{A}$ and by
    \autoref{lem:machine_monotone}
    $\sigma = \lock\eta_N\tick\eta_L$. Hence,
    $\adv\,l \in \einterp{\Str
      A}{\lock\eta_L\tick}{\set{\emptyset}^{k}}$. That is,
    $\state{t}{\eta} \forwards{v} \state{\adv\,l}{\eta_L}$ and
    $\state{\adv\,l}{\eta_L} \in T_k(A)$.
  \end{enumerate}
  \vspace{-1.5em}
\end{proof}

For the proof of causality we need the following property of the
operational semantics, which essentially states that we never read
from the `later' heap. 
\begin{lemma}
  \label{lem:dontTouchNow}
  If $\hevalSingle{t}{\sigma,l\mapsto u}{v}{\sigma',l\mapsto u}$, then 
  $\hevalSingle{t}{\sigma,l\mapsto u'}{v}{\sigma',l\mapsto u'}$ for
  any $u'$.
\end{lemma}
To prove the above lemma we have to make the following reasonable
assumption about the function $\allocate{\cdot}$ that performs the
allocation of fresh heap locations: Given two stores $\sigma, \sigma'$
with $\dom{\later{\sigma}} = \dom{\later{\sigma'}}$, we have that
$\allocate{\sigma} = \allocate{\sigma'}$. In other words,
$\allocate{\sigma}$ only depends on the domain of the `later' heap. For
example, if the set of heap locations is just $\nats$, then
$\allocate{\sigma}$ could be implemented as the smallest heap location
that is fresh for the `later' heap of $\sigma$.

Analogously to the family of sets $S_k(A)$, we construct a family of
sets $T_k(A,B)$ that contains all states which are safe to run for $k$
more steps on the stream transducer semantics:
\[
  T_k(A,B) = \setcom{\pair{t}{\eta}} {\thel \nin \dom{\eta}
    \land \forall v,w \in \vinterp{A}{\bot}{()}. t \in \einterp{\Str B}{\lock
      \eta,\thel \mapsto v :: \thel\tick \thel \mapsto w
      :: \thel}{H^k(A)}}
\]
where
$H(A) = \setcom{\thel \mapsto v :: \thel}{v \in
  \vinterp{A}{\bot}{()}}$ and $H^k(A)$ is the sequence of $k$ copies
of $H(A)$.

Finally, we give the proof of causality:
\begin{proof}[Proof of \autoref{thr:causality}]~
    \begin{enumerate}[(i)]
    \item Given $\hastype{}{t}{\Box(\Str A \to \Str B)}$ and
      $v,v' \in \vinterp{A}{\bot}{()}$, we need to show that
      $\unbox\, t\,(\adv\, \thel) \in \einterp{\Str B}{\lock \thel
        \mapsto v :: \thel \tick \thel \mapsto v' :: \thel}{H(A)^k}$.
      By induction on $k+1$ we can show that
      $\thel \in \vinterp{\Delay\Str A}{\lock \thel \mapsto v ::
        \thel}{H^{k+1}(A)}$. By definition of the value relation, this
      means that
      $v :: \thel \in \einterp{\Str A}{\lock \thel \mapsto v ::
        \thel\tick \thel \mapsto v' :: \thel}{H^k(A)}$, which in turn
      implies that
      $\adv\,\thel \in \einterp{\Str A}{\lock \thel \mapsto v ::
        \thel\tick \thel \mapsto v' :: \thel}{H^k(A)}$. Since
      $\hastype{}{t}{\Box(\Str A \to \Str B)}$, we know that
      $\hastype{\lock}{\unbox\, t}{\Str A \to \Str B}$.  Using
      \autoref{thr:lrl} we thus obtain that
      $\unbox\,t \in \einterp{\Str A \to \Str B}{\lock}{H^k(A)}$,
      which in turn implies that
      $\unbox\,t \in \einterp{\Str A \to \Str B}{\lock \thel \mapsto v
        :: \thel \tick \thel \mapsto v' :: \thel}{H^k(A)}$. Therefore,
      we have that
      $\unbox\, t\,(\adv\, \thel) \in \einterp{\Str B}{\lock \thel
        \mapsto v :: \thel \tick \thel \mapsto v' :: \thel}{H^k(A)}$.
    \item Let $\pair{t}{\eta} \in T_{k+1}(A,B)$ and
      $\hastype{}{v}{A}$. We need to find $l, \eta_N, \eta_L$, and
      $\hastype{}{v'}{B}$ such that
        \begin{align}
          &\hevalSingle{t}{\lock \eta , \thel\mapsto v :: \thel \tick
            \thel\mapsto \unit}{v' ::
            l}{\lock\eta_N\tick\eta_L,\thel\mapsto\unit}\tag{1}
          \label{eq:causality1}
          \\\text{ and }\quad
          &\adv\, l \in \einterp{\Str B}{\lock \eta_L,\thel\mapsto w ::
            \thel\tick\thel\mapsto w' :: \thel}{H^{k}(A)}
          \quad\text{ for all } w,w' \in \vinterp{A}{\bot}{()}.
          \tag{2}
          \label{eq:causality2}
        \end{align}
    By \autoref{lem:vinterp_value}~\ref{item:vinterp_value1},
    $v \in \vinterp{A}{\bot}{()}$, and by
    \autoref{lem:vinterp_value}~\ref{item:vinterp_value2}, there is
    some $w^* \in \vinterp{A}{\bot}{()}$. Since
    $t \in \einterp{\Str B}{\lock \eta,\thel \mapsto v :: \thel\tick
      \thel \mapsto w^* :: \thel}{H^{k+1}(A)}$, we have that
    \[
      \hevalSingle{t}{\lock \eta , \thel\mapsto v :: \thel \tick\thel\mapsto w^* :: \thel}{v''}{\sigma} \text{ and } v'' \in
      \vinterp{\Str B}{\sigma}{H^{k+1}(A)}
    \]
    Consequently, $v'' = v' :: l$ for some
    $v'\in \vinterp{B}{\sigma}{H^{k+1}(A)}$ by
    \autoref{lem:machine_monotone}, $\sigma$ is of the form
    $\lock\eta_N\tick\eta_L,\thel\mapsto w^* :: \thel$.  By
    \autoref{lem:dontTouchNow} and \autoref{lem:vinterp_value}~\ref{item:vinterp_value1}, we
    then have \eqref{eq:causality1} and $\hastype{}{v'}{B}$,
    respectively.
    
    Finally, to prove \eqref{eq:causality2}, we assume
    $w,w' \in \vinterp{A}{\bot}{()}$ and show
    $\adv\, l \in \einterp{\Str B}{\lock \eta_L,\thel\mapsto w ::
      \thel\tick\thel\mapsto w' :: \thel}{H^{k}(A)}$.  Since
    $t \in \einterp{\Str B}{\lock \eta,\thel \mapsto v :: \thel\tick
      \thel \mapsto w :: \thel}{H^{k+1}(A)}$, we have that
    \[
      \hevalSingle{t}{\lock \eta , \thel\mapsto v :: \thel \tick\thel\mapsto w :: \thel}{v'''}{\sigma'} \text{ and } v''' \in
      \vinterp{\Str B}{\sigma'}{H^{k+1}(A)}
    \]
    By \autoref{lem:dontTouchNow} and \autoref{prop:determ} we thus
    know that $v''' = v' :: l$ and
    $\sigma' = \lock\eta_N\tick\eta_L,\thel\mapsto w ::
    \thel$. Consequently,
    $\sigma'(l) \in \einterp{\Str B}{\lock\eta_L,\thel\mapsto w ::
      \thel\tick\thel\mapsto w' :: \thel}{H^{k}}$, which implies that
    \[
      \adv\, l \in \einterp{\Str B}{\lock\eta_L,\thel\mapsto w ::
        \thel\tick\thel\mapsto w' :: \thel}{H^{k}(A)}.
    \]
  \end{enumerate}
\end{proof}

\section{Related Work}
\label{sec:related-work}

The central ideas of functional reactive programming were originally
developed for the language Fran~\cite{FRAN} for reactive
animation. These ideas have since been developed into general purpose
libraries for reactive programming, most prominently the Yampa
library~\citep{nilsson2002} for Haskell, which has been used in a
variety of applications including games, robotics, vision, GUIs, and
sound synthesis. Some of these libraries use a continuous notion of
time, allowing e.g., integrals over input data to be computed. A
continuous notion of time can be encoded in Simply RaTT as well given
that the language is extended with a type $\Time$ that suitably
represents positive time intervals (e.g., floating-point numbers). For
example, a Yampa-style signal function type from $A$ to $B$ is thus
encoded as $\Box (\Str\Time \to \Str A \to \Str
B)$. This encoding reflects the (unoptimised) definition of
Yampa-style signal functions~\citep{nilsson2002}, which is a
coinductive type satisfying
$\sym{SF}\; A\; B \cong \Time \to A \to (B \times \sym{SF}\;A\;
B)$. We believe that it should be possible to implement a Yampa-style
FRP library in Simply RaTT, and \autoref{sec:stream_library} has some
examples of combinators similar to those found in Yampa. While some of
these combinators have stability constraints on types, we believe that
these constraints will always be satisfied in concrete applications.

Simply RaTT follows a \emph{pull}-based approach to FRP, which means
that the program is performing computations at every time step even if
no event occurred. \citet{push-pull} proposed an implementation of an
FRP library that combines \emph{pull}-based FRP with a
\emph{push}-based approach, where computation is only performed in
response to incoming events. Whereas a pull-based approach is
appropriate for example in games, which run at a fixed sampling rate,
a push-based approach is more efficient for applications like GUIs,
which often only need to react to events that occur infrequently.

The idea of using modal type operators for reactive programming goes
back at least to the independent works of \citet{jeffrey2012,krishnaswami2011ultrametric} and
\citet{jeltsch2013temporal}. One of the inspirations for \citet{jeffrey2012} was to
use linear temporal logic~\citep{ltl} as a programming language
through the Curry-Howard isomorphism. The work of Jeffrey and Jeltsch
has mostly been based on denotational semantics, and
\citet{krishnaswami2011ultrametric,
  krishnaswami2012higher,krishnaswami13frp,cave14fair} are the only
works to our knowledge giving operational guarantees. The work of
\citet{cave14fair} shows how one can encode notions of fairness in
modal FRP, if one replaces the guarded fixed point operator with more
standard (co)recursion for (co)inductive types.
 
The guarded recursive types and fixed point combinator originate with \citet{nakano2000},
but have since been used for constructing logics for reasoning about advanced programming languages
~\cite{ToT} using an abstract form of step-indexing~\cite{appel01indexed}. The Fitch-style approach
to modal types~\citep{Fitch:Symbolic} has been used for guarded recursion in Clocked Type Theory~\cite{bahr2017clocks}, 
where contexts 
can contain multiple, named ticks. Ticks can be used for reasoning about 
guarded recursive programs. 
The denotational semantics of Clocked Type Theory~\cite{CloTTModel} reveals the difference
from the more standard two-context approaches to modal logics, such as Dual Intuitionistic Linear Logic
\cite{barber1996dual}: In the latter, the modal operator is implicitly applied to the type of all variables in one context,
in the Fitch-style, placing a tick in a context corresponds to applying a \emph{left adjoint} to the modal
operator to the context.

The previous work closest to the present work is that of \citet{krishnaswami13frp}. 
We have already compared to this several times above, but give a short summary
here. Simply RaTT is expressive enough to encompass all the positive examples
of Krishnaswami's calculus, but we go a step further and identify a source of time leaks
which allows us to eliminate in typing a number of leaking examples typable in Krishnaswami's
calculus including the $\Varid{leakyNats}$ example from the introduction, and $\Varid{scary\_const}$. 
One might claim that these are explicit leaks, but detecting them in the type system is 
a major step forward we believe. Note that the Fitch-style approach is a real shift in
approach: The time dependencies have changed, and Krishnaswami's context
of stable variables has been replaced by a context of initial variables. 
One difference between these is that variables can be introduced 
from Krishnaswami's stable context. In Simply RaTT, initial variables
can generally not be introduced into temporal judgements. We plan to
explain this change in terms of denotational semantics in future work.

Another approach to reactive programming is that of synchronous
dataflow languages.  Here the main abstraction is that of a ``logical
tick'' or synchronous abstraction.
This is the assumption that at each tick, the output is computed
instantaneously from the input.
This abstraction makes reasoning about time much easier than if we had
to consider both the reactive behaviour and the internal timing
behaviour of a program.  Of particular interest is the synchronous
dataflow language Lustre \cite{caspi1987lustre}.  Lustre is a
first-order language used for describing and verifying real-time
systems and is at the core of the SCADE industrial environment
\cite{scadeWebsiteBackground} which is used for critical control
systems in aerospace, rail transportation, industrial systems and
nuclear power plants \cite{scadeWebsiteSuccess}.  In
\autoref{sec:lustre}, we have shown how to encode some of the simpler
concepts of Lustre in Simply RaTT, and how the concept of a logical
tick fits well with the notion of stepwise stream unfolding. 

\section{Conclusions and Future Work}
\label{sec:concl-future-work}

We have presented the modal calculus Simply RaTT for reactive programming. Using
the Fitch-style approach to modal types this gives a significant simplification of
the type system and programming examples over existing approaches, in particular
the calculus of \citet{krishnaswami13frp}. Moreover, we have identified a source of 
time leaks and designed the type system to rule these out. 

In future work we aim to extend Simply RaTT to a full type theory with dependent 
types for expressing properties of programs. Before doing that, however, we would
like to extend Simply RaTT to encode fairness in types as in the work of \citet{cave14fair}.
This is not easy, since it requires a distinction between inductive and coinductive 
guarded types, but Nakano's fixed point combinator forces these to coincide.

\begin{acks}                     
This work was supported by a research grant (13156) from VILLUM FONDEN.

\end{acks}

\bibliography{paper}



\end{document}

%% file: prelude.tex
\usepackage{amsmath,amssymb,amsfonts,amsthm,stmaryrd}
\usepackage{centernot}
\usepackage{xspace}
\usepackage{mathpartir}
\usepackage{natbib}
\usepackage[shortlabels,inline]{enumitem}

\usepackage{hyperref}
\usepackage{tikz-cd}
\usepackage{bussproofs}
\EnableBpAbbreviations 
\usepackage{color}
\usepackage{algorithmicx}
\usepackage{algpseudocode}
\usepackage{algorithm}
\usepackage{epigraph}
\usepackage{fancyvrb}

\usepackage{thmtools}

\DefineVerbatimEnvironment%
{RaTT}{Verbatim}
{fontfamily=helvetica,numbers=right,commandchars=\\\{\},frame=single}

\newcommand\nin{\not\in}
\newcommand\sym[1]{\mathsf{#1}}
\newcommand\pair[2]{\left\langle#1,#2\right\rangle}
\newcommand\set[1]{\left\{#1\right\}}
\newcommand{\setcom}[2]{\set{#1\left\vert\vphantom{#1}\,#2\right.}}
\let \toold \to
\renewcommand{\to}{\ensuremath{\toold}}
\let \lambdaold \lambda
\renewcommand{\lambda}{\ensuremath{\lambdaold}}
\newcommand{\delayapp}{\ensuremath{\circledast}}
\newcommand{\progressapp}{\ensuremath{\odot}}
\newcommand{\boxapp}{\ensuremath{\boxast}}
\newcommand{\promoteapp}{\ensuremath{\boxdot}}
\newcommand{\world}[2]{(#1 \diamond #2)}

\newcommand\allocate[1]{\sym{alloc}\left(#1\right)}

\newcommand\dom[1]{\sym{dom}\left(#1\right)}
\newcommand\ol[1]{\overline{#1}}
\newcommand\tickle{\sqsubseteq_\tick}
\newcommand\tickge{\sqsupseteq_\tick}
\newcommand\heaple{\sqsubseteq}
\newcommand\heapge{\sqsupseteq}
\newcommand\suffix{\le_{\sym{suf}}}

\newcommand{\interm}{\sym{in}}
\newcommand{\caseterm}[5]{\sym{case}\, #1 \, \sym{of}\, \interm_1\, #2 . #3 ; \interm_2\,#4 . #5}

\newcommand\unit{\ensuremath{\langle\rangle}}

\newcommand\fix{\sym{fix}}
\newcommand\tick{\checkmark}
\newcommand\lock{\sharp}
\newcommand\adv{\sym{adv}}
\newcommand\delay{\sym{delay}}
\newcommand\rbox{\sym{box}}
\newcommand\promote{\sym{promote}}
\newcommand\progress{\sym{progress}}
\newcommand\unbox[1][]{\sym{unbox}^{#1}}
\newcommand\out{\sym{out}}
\newcommand\into{\sym{into}}
\newcommand{\head}{\sym{head}}
\newcommand{\tail}{\sym{tail}}

\newcommand{\later}[1]{\sym{later}(#1)}

\newcommand{\trim}[1]{\sym{gc}(#1)}
\newcommand{\just}{\sym{just}}
\newcommand{\nothing}{\sym{nothing}}
\newcommand{\val}{\sym{val}}
\newcommand{\wait}{\sym{wait}}

\let \Boxold \Box
\renewcommand{\Box}{\ensuremath{\Boxold}}
\newcommand\Delay{\ensuremath{{\bigcirc}}}

\newcommand\stable{\sym{stable}}
\newcommand\Str[1]{\ensuremath{\sym{Str}(#1)}}
\newcommand\Ev[1]{\ensuremath{\sym{Ev}(#1)}}

\newcommand\Nat{\ensuremath{\sym{Nat}}}
\newcommand\Time{\ensuremath{\sym{Time}}}
\newcommand\Unit{\ensuremath{\sym{1}}}
\newcommand\Maybe[1]{\ensuremath{\sym{Maybe}(#1)}}

\newcommand\nats{{\mathbb N}}

\newcommand\vinterp[3]{\mathcal{V}\llbracket #1 \rrbracket_{#2}^{#3}}
\newcommand\einterp[3]{\mathcal{T}\llbracket #1 \rrbracket_{#2}^{#3}}
\newcommand\cinterp[3]{\mathcal{C}\llbracket #1 \rrbracket_{#2}^{#3}}

\newcommand\sem[1]{\llbracket #1 \rrbracket}

\newcommand\wfcxt[2][]{#2 \vdash_{#1} }
\newcommand\type{\sym{type}}
\newcommand\hastype[3]{#1 \vdash #2 : #3}
\newcommand\istype[2]{#1 \vdash #2: \type}
\newcommand\tokenfree[1]{\sym{token}\text{-}\sym{free}(#1)}
\newcommand\tickfree[1]{\sym{tick}\text{-}\sym{free}(#1)}

\renewcommand\state[2]{\left\langle #1;#2 \right\rangle}
\newcommand\hevalSingle[4]{ \state{#1}{#2} \Downarrow \state{#3}{#4}}

\newcommand\thel{l^*}
\newcommand\forward[2]{\stackrel{#1/#2}{\Longrightarrow}}
\newcommand\forwards[1]{\stackrel{#1}{\Longrightarrow}}

\newcommand\wlen[1]{| #1|}
\newcommand\tsize[1]{\left| #1\right|}

\newcommand{\citepos}[1]{\citeauthor{#1}'s [\citeyear{#1}]}
